\documentclass[11pt]{article}
\title{Sublinear Algorithms and Lower Bounds for Metric TSP Cost Estimation} %TODO Please add

\author{Yu Chen\thanks{Department of Computer and Information Science, University of Pennsylvania. \newline ~\indent~~Email: {{\small {\tt \{chenyu2,kannan,sanjeev\}@cis.upenn.edu.}}}}
\and Sampath Kannan\samethanks
\and Sanjeev Khanna\samethanks
}

\usepackage{amsthm}
\usepackage{graphicx} % support the \includegraphics command and options
\usepackage{array} % for better arrays (eg matrices) in maths

\usepackage{amsmath, amssymb, amsfonts, verbatim}
\usepackage{hyphenat,epsfig,multirow}

\usepackage[usenames,dvipsnames]{xcolor}
\usepackage[ruled]{algorithm2e}

\usepackage[margin=1in]{geometry}

\usepackage{tcolorbox}
\tcbuselibrary{skins,breakable}
\tcbset{enhanced jigsaw}

\definecolor{DarkRed}{rgb}{0.5,0.1,0.1}
\definecolor{DarkBlue}{rgb}{0.1,0.1,0.5}

\usepackage{bm}
\usepackage{url}
\usepackage{xspace}

\usepackage{mdframed}

\usepackage[noend]{algpseudocode}
\makeatletter
\def\BState{\State\hskip-\ALG@thistlm}
\makeatother

\usepackage{cite}
\usepackage{enumitem}

\newtheorem{theorem}{Theorem}
\newtheorem{lemma}{Lemma}[section]
\newtheorem{proposition}[lemma]{Proposition}
\newtheorem{corollary}[theorem]{Corollary}
\newtheorem{claim}[lemma]{Claim}

\newtheorem{definition}{Definition}

\newtheorem{problem}{Problem}

\newtheorem*{claim*}{Claim}
\newtheorem*{proposition*}{Proposition}
\newtheorem*{lemma*}{Lemma}
\newtheorem*{problem*}{Problem}
\newtheorem*{theorem*}{Theorem}

\newtheorem{mdresult}{Result}

\newtheorem{mdinvariant}{Invariant}

\theoremstyle{definition}
\newtheorem{remark}[lemma]{Remark}

\allowdisplaybreaks

\newcommand{\hide}[1]{}

\usepackage{tikz}
\usetikzlibrary{arrows}
\usetikzlibrary{arrows.meta}
\usetikzlibrary{shapes}
\usetikzlibrary{backgrounds}
\usetikzlibrary{positioning}
\usetikzlibrary{decorations.markings}
\usetikzlibrary{patterns}
\usetikzlibrary{calc}
\usetikzlibrary{fit}
\usetikzlibrary{snakes}
\usetikzlibrary{shadows.blur}
\usetikzlibrary{petri,decorations.markings}

\usepackage{caption}
\usepackage{subcaption}

\renewcommand{\qed}{\nobreak \ifvmode \relax \else
      \ifdim\lastskip<1.5em \hskip-\lastskip
      \hskip1.5em plus0em minus0.5em \fi \nobreak
      \vrule height0.75em width0.5em depth0.25em\fi}

\newcommand{\EE}{\mathcal{E}}
\newcommand{\Alg}{\mathcal{A}}
\newcommand{\dis}{\mathcal{D}}

\newcommand{\cei}[1]{\lceil #1 \rceil}
\newcommand{\DYS}{\dis_{YES}}
\newcommand{\DNS}{\dis_{NO}}
\newcommand{\Equ}{\text{E3LIN2}}

\definecolor{DarkRed}{rgb}{0.5,0.1,0.1}
\definecolor{DarkBlue}{rgb}{0.1,0.1,0.5}

\colorlet{YellowOrange}{RawSienna}

\usepackage{hyperref}%[backref=page]
\hypersetup{
colorlinks=true,
pdfnewwindow=true,
citecolor=ForestGreen,
linkcolor=DarkRed,
filecolor=DarkRed,
urlcolor=DarkBlue
}

\newcommand*\samethanks[1][\value{footnote}]{\footnotemark[#1]}

\newcommand{\toShrink}{-.20cm}
\newcommand{\toShrinkEnu}{-.2cm}

\newcommand{\eps}{\ensuremath{\varepsilon}}

\newcommand{\bracket}[1]{\left[#1\right]}

\newcommand{\card}[1]{\left\vert{#1}\right\vert}

\newcommand{\ceil}[1]{{\left\lceil{#1}\right\rceil}}

\newcommand{\expect}[1]{\Exp\bracket{#1}}

\newcommand{\set}[1]{\ensuremath{\left\{ #1 \right\}}}

\DeclareMathOperator*{\Exp}{\ensuremath{{\mathbb{E}}}}

\newenvironment{tbox}{\begin{tcolorbox}[
		enlarge top by=5pt,
		enlarge bottom by=5pt,
		 breakable,
		 boxsep=0pt,
                  left=4pt,
                  right=4pt,
                  top=10pt,
                  boxrule=1pt,toprule=1pt,
                  colback=white,
                  arc=-1pt,
                  ]%%
	}
{\end{tcolorbox}}

\date{}

\begin{document}

\maketitle

\thispagestyle{empty}
\begin{abstract}
We consider the problem of designing sublinear time algorithms for estimating the cost of minimum metric traveling salesman (TSP) tour. Specifically, given access to a $n \times n$ distance matrix $D$ that specifies pairwise distances between $n$ points, the goal is to estimate the TSP cost by performing only sublinear (in the size of $D$) queries. For the closely related problem of estimating the weight of a metric minimum spanning tree (MST), it is known that for any $\eps > 0$, there exists an $\tilde{O}(n/\eps^{O(1)})$ time algorithm that returns a $(1 + \eps)$-approximate estimate of the MST cost. This result immediately implies an $\tilde{O}(n/\eps^{O(1)})$ time algorithm to estimate the TSP cost to within a $(2 + \eps)$ factor for any $\eps > 0$. However, no $o(n^2)$ time algorithms are known to approximate metric TSP to a factor that is strictly better than $2$. On the other hand, there were also no known barriers that rule out existence of $(1 + \eps)$-approximate estimation algorithms for metric TSP with $\tilde{O}(n)$ time for any fixed $\eps > 0$. In this paper, we make progress on both algorithms and lower bounds for estimating metric TSP cost.

On the algorithmic side, we first consider the graphic TSP problem where the metric $D$ corresponds to shortest path distances in a connected unweighted undirected graph. We show that there exists an $\tilde{O}(n)$ time algorithm that estimates the cost of graphic TSP to within a factor of $(2-\eps_0)$ for some $\eps_0>0$. This is the first sublinear cost estimation algorithm for graphic TSP that achieves an approximation factor less than $2$. We also consider another well-studied special case of metric TSP, namely, $(1,2)$-TSP where all distances are either $1$ or $2$, and give an $\tilde{O}(n^{1.5})$ time algorithm to estimate optimal cost to within a factor of $1.625$. Our estimation algorithms for graphic TSP as well as for $(1,2)$-TSP naturally lend themselves to $\tilde{O}(n)$ space streaming algorithms that give an $11/6$-approximation for graphic TSP and a $1.625$-approximation for $(1,2)$-TSP.
These results motivate the natural question if analogously to metric MST, for any $\eps > 0$, $(1 + \eps)$-approximate estimates can be obtained for graphic TSP and $(1,2)$-TSP using  $\tilde{O}(n)$ queries. We answer this question in the negative -- there exists an $\eps_0 > 0$, such that any algorithm that estimates the cost of graphic TSP ($(1,2)$-TSP) to within a $(1 + \eps_0)$-factor, necessarily requires $\Omega(n^2)$ queries. This lower bound result highlights a sharp separation between the metric MST and metric TSP problems. 

Similarly to many classical approximation algorithms for TSP, our sublinear time estimation algorithms utilize subroutines for estimating the size of a maximum matching in the underlying graph. We show that this is not merely an artifact of our approach, and that for any $\eps > 0$, any algorithm that estimates the cost of graphic TSP or $(1,2)$-TSP to within a $(1 + \eps)$-factor, can also be used to estimate the size of a maximum matching in a bipartite graph to within an $\eps n$ additive error. This connection allows us to translate known lower bounds for matching size estimation in various models to similar lower bounds for metric TSP cost estimation.
\end{abstract}

\setcounter{page}{0}
\clearpage

\section{Introduction}

In the metric traveling salesman problem (TSP), we are given $n$ points in an arbitrary metric space with an $n \times n$ matrix $D$ specifying pairwise distances between them. The goal is to find a simple cycle (a TSP tour) of minimum cost that visits all $n$ points. An equivalent view of the problem is that we are given a complete weighted undirected graph $G(V,E)$ where the weights satisfy triangle inequality, and the goal is to find a Hamiltonian cycle of minimum weight. The study of metric TSP is intimately connected to many algorithmic developments, and the polynomial-time approximability of metric TSP and its many natural variants are a subject of extensive ongoing research (see, for instance,~\cite{Vygen2012_Survey, sebo2014shorter, AnKS15, KarpinskiLS15, momke2016removing, SeboZ16, Gao18, TraubV18, MnichM18, TraubV19} and references within for some relatively recent developments). In this paper, we consider the following question: can one design sublinear algorithms that can be used to obtain good estimates of the cost of an optimal TSP tour? Since the complete description of the input metric is of size $\Theta(n^2)$, the phrase sublinear here refers to algorithms that run in $o(n^2)$ time. 

A standard approach to estimating the metric TSP cost is to compute the cost of a minimum spanning tree (MST), and output two times this cost as the estimate of the TSP cost (since any spanning tree can be used to create a spanning simple cycle by at most doubling the cost). 
The problem of approximating the cost of the minimum
spanning tree in sublinear time was first studied in the graph adjacency-list model by Chazelle, Rubinfeld, and Trevisan~\cite{ChazelleRT05}. The authors gave an $\tilde{O}(d W /\eps^2)$-time algorithm to estimate the MST cost to within a $(1+\eps)$-factor in graphs where average degree is $d$, and all edge costs are integers in $[1..W]$. For certain parameter regimes this gives a sublinear time algorithm for estimating the MST cost but in general, this run-time need not be sublinear.
Subsequently, in an identical setting as ours, Czumaj and Sohler~\cite{czumaj2009estimating} showed that for any $\eps > 0$, there exists an $\tilde{O}(n/\eps^{O(1)})$ time algorithm that returns a $(1 + \eps)$-approximate estimate of the MST cost when the input is an $n$-point metric. This result immediately implies an $\tilde{O}(n/\eps^{O(1)})$ time algorithm to estimate the TSP cost to within a $(2 + \eps)$ factor for any $\eps > 0$. However, no $o(n^2)$ query algorithms are known to approximate metric TSP to a factor that is strictly better than $2$. On the other hand, there are also no known barriers that rule out existence of $(1 + \eps)$-approximate estimation algorithms for metric TSP with $\tilde{O}(n)$ queries for any fixed $\eps > 0$. In this paper, we make progress on both algorithms and lower bounds for estimating metric TSP cost.

On the algorithmic side, we first consider the {\em graphic TSP} problem, an important case of metric TSP that has been extensively studied in the classical setting -- the metric $D$ corresponds to the shortest path distances in a connected unweighted undirected graph~\cite{momke2016removing,mucha2014frac,sebo2014shorter}. We give the first $\tilde{O}(n)$ time algorithm for graphic TSP that achieves an approximation factor {\em strictly better} than 2.

\begin{theorem}
\label{thm:graphic_tsp_linear}
There is an $\tilde{O}(n)$ time randomized algorithm that estimates the cost of graphic TSP to within a factor of $2-\eps_0$ for some constant $\eps_0>0$.
\end{theorem}

On the other hand, if we are willing to allow a higher sublinear time, we can get a better approximation ratio.

\begin{theorem}
\label{thm:graphic_tsp_1.5}
There is an $\tilde{O}(n^{1.5})$ time randomized algorithm that estimates the cost of graphic TSP to within a factor of $(27/14)$. 
\end{theorem}

At a high-level, our algorithm is based on showing the following: if a graph $G$ either lacks a matching of size $\Omega(n)$ or has $\Omega(n)$ biconnected components (blocks), then the optimal TSP cost is not too much better than $2n$. Note that a connected unweighted instance of graphic TSP always contains a TSP tour of cost at most $2n$ since the MST cost is $(n-1)$ on such instances.  Conversely, if the graph $G$ has both a large matching and not too many blocks, then we can show that the optimal TSP cost is distinctly better than $2n$. Since we do not know an efficient sublinear algorithm to estimate the number of blocks in a graph $G$, we work with another quantity that serves as a proxy for this and can be estimated in $\tilde{O}(n)$ time. The main remaining algorithmic challenge then is to estimate sufficiently well the size of a largest matching. This problem is very important by itself, and has received much attention~\cite{PR07,NguyenO08,yoshida2012improved,onak2012near,kapralov2020space}; please see a detailed discussion of this problem, and relevant recent developments towards the end of this section. Our $\tilde{O}(n)$ query results utilize the recent result of Kapralov {\em et al.} \cite{kapralov2020space} who give an algorithm to approximate the size of maximum matching to within a constant factor (for some very large constant) in $\tilde{O}(n)$ time in the {\em pair query} model (is there an edge between a given pair of vertices?). We also show that matching size can be estimated to within a factor of $2$ in $\tilde{O}(n^{1.5})$ time, crucial to obtaining the approximation guarantee in Theorem~\ref{thm:graphic_tsp_1.5}.

Our approach for estimating graphic TSP cost in sublinear time also lends itself to an $\tilde{O}(n)$ space streaming algorithm that can obtain an even better estimate of the cost. To our knowledge, no estimate better than a $2$-approximation was known previously. In the streaming model, we assume that the input to graphic TSP is presented as a sequence of edges of the underlying graph $G$. Any algorithm for this model, clearly also works if instead the entries of the distance matrix are presented in the stream -- an entry that is $1$ corresponds to an edge of $G$, and it can be ignored otherwise as a non-edge. 

\begin{theorem}
\label{thm:graphic_tsp_stream_main}
There is an $O(n)$ space randomized streaming algorithm that estimates the cost of graphic TSP to within a factor of $(11/6)$ in insertion-only streams. 
\end{theorem}

We next consider another well-studied special case of metric TSP, namely, $(1,2)$-TSP where all distances are either $1$ or $2$~\cite{adamaszek2018new,berman20068,papadimitriou1993traveling}, and obtain the following result.

\begin{theorem}
\label{thm:1,2_tsp_main}
There is an $\tilde{O}(n^{1.5})$ time randomized algorithm that estimates the cost of $(1,2)$-TSP  to within a factor of $1.625$. 
\end{theorem}

Throughout the paper, whenever we refer to a graph associated with a $(1,2)$-TSP instance, it refers to the graph $G$ induced by edges of distance $1$ in our $\{1,2\}$-metric. At a high-level, the idea underlying our algorithm is to analyze the structure of the graph $G$ induced by edges of distance $1$. 
We design an algorithm to estimate the size of a maximal ``matching pair'' of $G$ which is defined to be the union of a pair of edge-disjoint matchings that is maximal, i.e., that is not a proper subset of another union of edge disjoint matchings. We show that whenever the size of a matching pair is large in a graph $G$, the TSP cost is distinctly smaller than $2n$, and conversely, if this quantity is not large, the TSP cost is close to $2n$. The main remaining algorithm challenge then is to estimate sufficiently well the size of a maximal matching pair, and we show that this can be done in $\tilde{O}(n^{1.5})$ time.

For $(1,2)$-TSP, an $\tilde{O}(n)$ query algorithm that estimates the cost of $(1,2)$-TSP to within a factor of $1.75$ was claimed in~\cite{sublinear71} but this result is based on the matching size estimation results of~\cite{onak2012near}. Unfortunately, as confirmed by the authors~\cite{onakpersonalcommunication}, there is a problem with the proof of one of the statements in the paper --- Observation 3.9 --- which is crucial for the correctness of the main result. 
As a result, the $\tilde{O}(d)$ time result in the neighbor query model as well as the $\tilde{O}(n)$ time result in the adjacency matrix, claimed in~\cite{onak2012near} can no longer be relied upon, and we have chosen to make this paper independent of these results. It is worth mentioning that if the $\tilde{O}(n)$-time matching estimation result of~\cite{onak2012near} can be shown to hold, then the run-time of both Theorems~\ref{thm:graphic_tsp_1.5} and~\ref{thm:1,2_tsp_main} can be improved to $\tilde{O}(n)$ time.

We note that it is easy to show that randomization is crucial to getting better than a $2$-approximation in sublinear time for both graphic TSP and $(1,2)$-TSP -- see Theorem~\ref{thm:det_lower_bound} in Section~\ref{sec:det}. The algorithms underlying Theorems~\ref{thm:graphic_tsp_1.5} and~\ref{thm:1,2_tsp_main}, lend themselves to $\tilde{O}(n)$ space single-pass streaming algorithms with identical approximation guarantees. 
These sublinear time algorithms motivate the natural question if analogously to metric MST, there exist sublinear time algorithms that for any $\eps > 0$, output a $(1 + \eps)$-approximate estimate of TSP cost for graphic TSP and $(1,2)$-TSP in $\tilde{O}(n)$ time. We rule out this possibility in a strong sense for both graphic TSP and $(1,2)$-TSP.

\begin{theorem}
\label{thm:ptas_lowerbound_main}
There exists an $\eps_0 > 0$, such that any randomized algorithm that estimates the cost of graphic TSP ($(1,2)$-TSP) to within a $(1 + \eps_0)$-factor, necessarily requires $\Omega(n^2)$ queries. 
\end{theorem}

This lower bound result highlights a sharp separation between the behavior of metric MST and metric TSP problems. At a high-level, our lower bound is inspired by the work of Bogdanov {\em et al.} \cite{bogdanov2002lower} who showed that any query algorithm that for any $\eps > 0$ distinguishes between instances of parity equations (mod $2$) that are either satisfiable (Yes) or at most $(1/2 + \eps)$-satisfiable (No), requires $\Omega(n)$ queries where $n$ denotes the number of variables. However, the query model analyzed in~\cite{bogdanov2002lower} is different from ours (see more details in Section~\ref{sec:ptas_lower_bound}). 
We first show that the lower bound of~\cite{bogdanov2002lower} can be adapted to an $\Omega(n^2)$ lower bound in our model, and then show that instances of parity equations can be converted into instances of graphic TSP (resp. $(1,2)$-TSP) such that for some $\eps_0 > 0$, any $(1 + \eps_0)$-approximation algorithm for graphic TSP (resp. $(1,2)$-TSP), can distinguish between the Yes and No instances of the parity equations, giving us the desired result.

Finally, similar to many classical approximation algorithms for TSP, our sublinear time estimation algorithms utilize subroutines for estimating the size of a maximum matching in the underlying graph. We show that this is not merely an artifact of our approach.

\begin{theorem}
\label{thm:tsp_vs_matchingsize_main}
For any $\eps \in [0,1/5)$, any algorithm that estimates the cost of an $n$-vertex instance of graphic TSP or $(1,2)$-TSP to within a $(1 + \eps)$-factor, can also be used to estimate the size of a maximum matching in an $n$-vertex bipartite graph to within an $\eps n$ additive error, with an identical query complexity, running time, and space usage.
\end{theorem} 

This connection allows us to translate known lower bounds for matching size estimation in various models to similar lower bounds for metric TSP cost estimation. In particular, using the results of~\cite{assadi2017estimating}, we can show that there exists an $\eps_0$ such that any randomized single-pass dynamic streaming algorithm for either graphic TSP or $(1,2)$-TSP that estimates the cost to within a factor of $(1+\eps_0)$, necessarily requires $\Omega(n^2)$ space.

We conclude by establishing several additional lower bound results that further clarify the query complexity of approximating TSP cost. For instance, we show that if an algorithm can access an instance of graphic TSP by only querying the edges of the graph (via neighbor and pair queries), then any algorithm that approximates the graphic TSP cost to a factor better than $2$, necessarily requires $\Omega(n^2)$ queries. This is in sharp contrast to Theorem~\ref{thm:graphic_tsp_linear}, and shows that working with the distance matrix is crucial to obtaining sublinear time algorithms for graphic TSP. We also show that even in the distance matrix representation, the task of {\em finding a  tour} that is $(2-\eps)$-approximate for any $\eps > 0$, requires $\Omega(n^2)$ queries for both graphic TSP and $(1,2)$-TSP.

\smallskip
\noindent
{\bf Matching Size Estimation:}
As the problem of matching size estimation is intimately connected to metric TSP cost estimation, we briefly review some relevant work here. This line of research primarily assumes that we are given a graph $G(V,E)$ with maximum degree $d$, that can be accessed via 
{\em neighbor queries}~\cite{GoldreichR02}: (a) for any vertex $v$, we can query its degree, and (b) for any vertex $v$ and an integer $i$, we can learn the $i^{th}$ neighbor of $v$. 

Parnas and Ron~\cite{PR07} initiated the study of matching size estimation in sublinear time and gave an $d^{O(\log(d/\eps)}$ time algorithm that estimates the matching size to within a constant factor plus an additive $\eps n$ error for any $\eps > 0$. Nguyen and Onak~\cite{NguyenO08} presented a new estimation algorithm and showed that it can estimate the matching size to within a factor of $2$ plus an additive $\eps n$ error in $2^{O(d)}/\eps^2$ time. We will refer to this approximation guarantee as a {\em $(2,\eps)$-approximation} of matching size. Yoshida {\em et al.}~\cite{yoshida2012improved} strongly improved upon the performance guarantee obtained in~\cite{NguyenO08}, and showed that a $(2,\eps)$-approximation to matching size can be accomplished in $O(d^4/\eps^2)$ time (in fact, they obtain the stronger $(2 \pm \eps)$-approximation guarantee). 
The analysis of~\cite{yoshida2012improved} was further improved by Onak {\em et al.}~\cite{onak2012near} who showed that the state of the art for $(2,\eps)$-approximation of matching size. 
We note that it is known that any $(O(1),\eps)$-approximate estimate of matching size necessarily requires $\Omega(d)$ queries~\cite{PR07}, so the result of~\cite{onak2012near} is essentially best possible. Unfortunately, as mentioned above, we recently discovered a subtle mistake in the analysis of Onak {\em et al.}~\cite{onakpersonalcommunication}. Consequently, the best known time complexity for obtaining a $(2,\eps)$-approximate estimate is $\tilde{O}(d^2/\eps^2))$; this weaker result also follows from the work of ~\cite{onak2012near}, but does not rely on the incorrect observation in~\cite{onak2012near}.

The difference between a linear dependence versus a quadratic dependence on degree $d$ is however huge in the sublinear time applications when the graph is not very sparse. In particular, while an $\tilde{O}(d)$ query result translates into an $\tilde{O}(n)$ time algorithm in the adjacency matrix model, an $\tilde{O}(d^2)$ query result gives only an $\tilde{O}(n^2)$ time algorithm, which is clearly not useful. Very recently, Kapralov {\em et al.}~\cite{kapralov2020space} gave an alternate approach based on a vertex ``peeling'' strategy (originally proposed in~\cite{PR07}) that yields an $(O(1),\eps)$-approximation of matching size in $\tilde{O}(d/\eps^2)$ time. Unfortunately, the constant hidden in the $O(1)$ notation is very large, and efficiently obtaining a $(2,\eps)$-approximation to matching size remains an important open problem. Meanwhile, by directly building on the work of~\cite{yoshida2012improved}, we obtain an $\tilde{O}(n^{1.5})$ time algorithm for a $(2,\eps)$-approximation to matching size in the adjacency matrix model, and it is this algorithm that is used in the results of Theorem~\ref{thm:graphic_tsp_1.5} and Theorem~\ref{thm:1,2_tsp_main}.

\smallskip
\noindent
{\bf Other Related Work:}
We note here that there is an orthogonal line of research that focuses on computing an approximate solution in near-linear time when the input is presented as a weighted undirected graph, and the metric is defined by shortest path distances on this weighted graph. It is known that in this model, for any $\eps > 0$, there is an $\tilde{O}(m/\eps^2 + n^{1.5}/\eps^3)$ time algorithm that computes a $(3/2 + \eps)$-approximate solution; here $n$ denotes the number of vertices and $m$ denotes the number of edges~\cite{CQ18}, and that a
$(3/2 + \eps)$-approximate estimate of the solution cost can be computed in $\tilde{O}(m/\eps^2 )$ time~\cite{CQ17}. It is not difficult to show that in this access model, even when the input graph is unweighted (i.e. a graphic TSP instance), any algorithm that outputs better than a $2$-approximate estimate of the TSP cost, requires $\Omega(n + m)$ time even when $m = \Omega(n^2)$. Hence this access model does not admit sublinear time algorithms that beat the trivial $2$-approximate estimate.

 \smallskip
 \noindent
 {\bf Organization:} In Section~\ref{sec:graph_TSP}, we present our algorithms for graphic TSP (Theorem~\ref{thm:graphic_tsp_linear}, Theorem~\ref{thm:graphic_tsp_1.5}, and Theorem~\ref{thm:graphic_tsp_stream_main}). In Section~\ref{sec:1,2_TSP}, we present the $1.625$-approximation algorithm of $(1,2)$-TSP (Theorem~\ref{thm:1,2_tsp_main}). In Section~\ref{sec:ptas_lower_bound}, we present our lower bound result that rules out possibility of a sublinear-time approximation scheme for both graphic TSP and $(1,2)$-TSP (Theorem~\ref{thm:ptas_lowerbound_main}). In Section~\ref{sec:tsp_vs_matchingsize}, we present a strong connection between approximating metric TSP cost and estimating matching size (Theorem~\ref{thm:tsp_vs_matchingsize_main}). Finally, in Section~\ref{sec:additional_lower_bounds}, we present several additional lower bound results on the complexity of approximating graphic TSP and $(1,2)$-TSP cost.

\section{Approximation for Graphic TSP Cost}
\label{sec:graph_TSP}

In this section, we exploit well-known properties of biconnected graphs and biconnected components in graphs to give an algorithm that achieves a $(2-\frac{1}{7c_0})$-approximation for graphic TSP if we have an efficient algorithm that approximates the maximum matching size within a factor of $c_0$. We first relate the cost of the TSP tour in a graph to the costs of the TSP tours in the biconnected components of the graph. Next we show that if the graph does not have a sufficiently big matching, it does not have a TSP tour whose length is much better than $2n$. We also show that if a graph has too many degree 1 vertices, or vertices of degree 2, both whose incident edges are bridges, then it does not have a TSP tour of cost much better than $2n$. We then establish the converse - a graph that has a good matching and not too many bad vertices (namely, vertices of degree $1$ or articulation points of degree $2$), then it necessarily has a TSP tour of cost much better than $2n$. We design $\tilde{O}(n)$ time test for the second condition, allowing us to approximate the cost of an optimal graphic TSP tour in sublinear time together with some known techniques for testing the first condition. In what follows, we first present some basic concepts and develop some tools that will play a central role in our algorithms.

\subsection{Preliminaries}
An unweighted graph $G = (V,E)$, defines a \textit{graphic metric} in $V$, where the distance between any two vertices $u$ and $v$ is given by the length of the shortest path between $u$ and $v$. The \textit{graphic TSP} is the Traveling Salesman Problem defined on such a graphic metric. In this paper our goal is to find a non-trivial approximation to the length of the traveling salesman tour in sublinear time in a model where we are allowed to make \textit{distance queries}.
In the distance query model, the algorithm can make a query on a pair of vertices $(u,v)$ and get back the answer $d(u,v)$, the distance between $u$ and $v$ in $G$. 

In a connected graph $G$, an edge $e$ is a \textit{bridge} if the deletion of $e$ would increase the number of connected components of $G$. A connected graph with no bridge is called a \textit{2-edge-connected graph}. A maximal 2-edge-connected subgraph of $G$ is called a \textit{2-edge-connected component}. The \textit{bridge-block tree} of a graph is a tree such that the vertex set contains the 2-edge-connected components and the edge set contains the bridges in the graph.

A connected graph $G$ is called \textit{2-vertex-connected} or \textit{biconnected} if when any one vertex is removed, the resulting graph remains connected. In a graph which is not biconnected, a vertex $v$ whose removal increases the number of components is called an \textit{articulation point}. It is easy to prove that any biconnected graph with at least $3$ vertices does not have degree $1$ vertices. A well-known alternate characterization of biconnectedness is that,
a graph $G$ is biconnected if and only if for any two distinct edges, there is a simple cycle that contains them.

A \textit{biconnected component} or \textit{block} in a graph is a maximal biconnected subgraph. Any graph $G$ can be decomposed into blocks such that the intersection of any two blocks is either empty, or a single articulation point. Each articulation point belongs to at least two blocks. %The size of graphic TSP of a graph $G$ is the sum of the size of graphic TSP of all its blocks. 
If a block is a single edge, then we call this block a \textit{trival block}; otherwise it is a \textit{non-trivial block}. A trival block is also a bridge in the graph. The size of a block is the number of vertices in the block.
The following lemma shows the relationship between the number of blocks and the sum of the sizes of the blocks.

\begin{lemma} \label{lem:blo-num}
    If a connected graph $G$ has $n$ vertices and $k$ blocks, then the sum of the sizes of the blocks is equal to $n + k-1$.
\end{lemma}

\begin{proof} 
    We prove the lemma by induction on the number $k$ of blocks. The base case is when $k=1$. In this case, $G$ itself is a block of size $n$.

    For the induction step, we have $k>1$ and thus the graph has at least one articulation point. Suppose $v$ is an arbitrary articulation point in $G$. Let $V_1,V_2,\dots, V_j$ be the set of vertices in the connected components of $G \setminus \{v\}$. We have $\sum_{i=1}^j \card{V_i} = n-1$. Let $G_1,G_2,\dots, G_j$ be the subgraphs of $G$ induced by $V_1 \cup \{v\}, V_2 \cup \{v\}, \dots, V_j \cup \{v\}$. For any $G_i$, let $k_i$ be the number of blocks in $G_i$, we have $\sum_{i=1}^j k_i = k$. By induction hypothesis, the sum of the sizes of blocks in $G_i$ is $\card{V_i}+1 + k_i -1 = \card{V_i}+k_i$. So the sum of the sizes of blocks in $G$ is $\sum_{i=1}^j \card{V_i} + k_i = n - 1 + k$.
\end{proof}

The block decomposition of a graph has a close relationship with the cost of graphic TSP of the graph.

\begin{lemma} [Lemma 2.1 of \cite{focs/MomkeS11}]\label{lem:blo-tsp}
    The cost of the graphic TSP of a connected graph $G=(V,E)$ is equal to the sum of the costs of the graphic TSP of all blocks in the graph.
\end{lemma}

Together these two lemmas give us a simple lower bound on the cost of the graphic TSP of a graph $G$ (using the fact that the cost of graphic TSP is at least the number of vertices in the graph).

\begin{lemma} \label{lem:blo-size}
    If a graph $G$ has $n$ vertices and $k$ blocks, then the cost of graphic TSP of $G$ is at least $n + k -1$.
\end{lemma}

An \textit{ear} in a graph is a simple cycle or a simple path. An ear which is a path is also called an \textit{open ear} and it has two endpoints, whereas for a cycle, one vertex is designated as the endpoint. An \textit{ear decomposition} of a graph is a partition of a graph into a sequence of ears such the endpoint(s) of each ear (except for the first) appear on previous ears and the internal points (the points that are not endpoints) are not on previous ears. A graph $G$ is biconnected if and only if $G$ has an ear decomposition such that each ear but the first one is an open ear \cite{west1996introduction}. An ear is \textit{nontrivial} if it has at least one internal point. The following lemma upper bounds the cost of graphic TSP of a biconnected graph.

\begin{lemma} [Lemma~5.3 of \cite{sebo2014shorter}, also a corollary of Lemma~3.2 of \cite{momke2016removing}] \label{lem:bi-tsp}
    Given a 2-vertex-connected graph $G=(V,E)$ and an ear-decomposition of $G$ in which all ears are nontrivial, a graphic TSP tour of cost at most $\frac{4}{3}(\card{V}-1)+ \frac{2}{3}\pi$ can be found in $O(\card{V}^3)$ time, where $\pi$ is the number of ears.
\end{lemma}

We now prove an important lemma that gives an upper bound on the cost of graphic TSP in a biconnected graph in terms of the size of a matching in the graph. 
\begin{lemma} \label{lem:bi-mat}
    Suppose $G$ is a biconnected graph with at least $n \ge 3$ vertices. If $G$ has a matching $M$, then the cost of graphic TSP of $G$ is at most $2n-2-\frac{2\card{M}}{3}$.
\end{lemma}

\begin{proof}
    We first find a spanning biconnected subgraph of $G$ that only contains $2n-2-M$ edges, then use Lemma~\ref{lem:bi-tsp} to bound the cost of graphic TSP.

    We construct a spanning biconnected subgraph $G^{\star} = P_0 \cup P_1 \cup \dots $ recursively: $P_0$ contains a single edge in $M$. If $G_{i-1} = P_0 \cup P_1 \cup \dots \cup P_{i-1}$ is a spanning subgraph of $G$, let $G^{\star}=G_{i-1}$ and finish the construction. Otherwise we construct $P_i$ as follows. Let $e$ be an edge in $M$  both whose endpoints are not in $G_{i-1}$. If there is no such edge, then let $e$ be an arbitrary edge such that at least one of its endpoints is not in $G_{i-1}$. Let $e'$ be an arbitrary edge in $G_{i-1}$. By the alternate characterization of biconnectedness, there is a simple cycle $C_i$ that contains both $e$ and $e'$. Let $P_i$ be the path in $C_i$ that contains $e$ and exactly two vertices in $G_{i-1}$, which are the endpoints of $P_i$.

    Since $P_i$ contains at least one vertex not in $G_{i-1}$, the construction always terminates. Note that $P_0 \cup P_1$ is a cycle, and each $P_i$ ($i>1$) is an open ear of $G^{\star}$. So, $(P_0 \cup P_1, P_2, \dots )$ is an open ear decomposition of $G^{\star}$, which means $G^{\star}$ is biconnected. 

    Now we prove that the number of edges in $G^{\star}$ is at most $2n-2-M$. Let $n_i$ be the number of vertices in $G_i \backslash G_{i-1}$. Let $G_{-1}$ be the empty graph, so that $n_0$=2. Let $p_i$ be the number of edges in $P_i$ and $m_i$ be the number of edges $e$ in $M$ such that $e \cap G_i \neq \emptyset$ and $e \cap G_{i-1}= \emptyset$. (Here we view an edge as a 2-vertex set.)  Note that $m_0 = 1$. Suppose $G^{\star} = G_k$. Then $\sum_{i=1}^k n_i = n$, $\sum_{i=1}^k p_i$ is the number of edges in $G^{\star}$ and $\sum_{i=1}^k m_i = \card{M}$. For any $i>0$, $P_i$ is an open ear whose internal points are not in $G_{i-1}$. So $n_i = p_i-1$. If there is an edge $e\in M$ such that $e \cap G_{i-1} = \emptyset$, then $P_i$ contains both endpoints of an edge in $M$, which means $m_i \le n_i -1$. If all edges in $M$ already have an endpoint in $G_{i-1}$, $m_i = 0 \le n_i-1$. So in both cases, $p_i = n_i + 1 = 2n_i - (n_i - 1) \le 2n_i - m_i$. Also, $p_0 = 1 = 2n_0 - 2 - m_0$. So the number of edges in $G^{\star}$ is $\sum_{i=0}^k p_i \le 2n_0 - 2 - m_0 + \sum_{i=1}^{k} (2n_i - m_i) = 2n - 2 - \card{M}$.

    Since $(P_0 \cup P_1, P_2, P_3, \dots, P_k)$ is an open ear decomposition of $G^{\star}$, the number of ears in $G$ is $k$. On the other hand, $\sum_{i=0}^k p_i = 1 + \sum_{i=1}^k (n_i+1) = n - 1 + k$, we have $n-1+k \le 2n - 2 - \card{M}$, which means $k \le n-1-\card{M}$. By Lemma~\ref{lem:bi-tsp}, the cost of graphic TSP of $G^{\star}$ is at most $\frac{4}{3} (n-1) + \frac{2}{3} k \le 2(n-1) - \frac{2}{3} \card{M}$. 
    
    Since $G^{\star}$ is a subgraph of $G$ that contains all the vertices in $G$, the cost of graphic TSP of $G$ is at most the cost of graphic TSP of $G^{\star}$, which is at most $2n - 2 -\frac{2}{3} \card{M}$.
\end{proof}

\subsection{Approximation Algorithm for Graphic TSP}
In this section, we give the algorithm that approximates the cost of graphic TSP of a graph $G$ within a factor of less than $2$. 

We call a vertex $v$ a \textit{bad} vertex if $v$ has degree $1$ or  is an articulation point with degree $2$.

For any given $\delta>0$, the graphic TSP algorithm performs the following two steps. 

\begin{enumerate}
    \item Obtain an estimate $\hat{\alpha} n$ of the size of maximum matching $\alpha n$. 
    \item Obtain an estimate $\hat{\beta} n$ of the number of bad vertices $\beta n$.
\end{enumerate}

The algorithm then output $\min \{2n,(2-\frac{2}{7}(\hat{\alpha} - 2\hat{\beta}))n \}$.

%\sk{Text here and lemma below should be replaced to not refer to Onak et al.} 

To perform the second step in $\tilde{O}(n)$ distance queries and time, we randomly sample $O(\frac{1}{\delta^2})$ vertices. For each sampled vertex, we can obtain the degree with $n$ queries. The following lemma shows that we can also check whether a degree $2$ vertex is an articulation point using distance queries in $O(n)$ time. Then by the Chernoff bound, we can approximate the number of bad vertices with additive error $O(\delta n)$ with a high constant probability.

\begin{lemma} \label{lem:deg2-art}
    Suppose a vertex $v$ in a connected graph $G$ has only two neighbors $u$ and $w$. The following three conditions are equivalent:
    \begin{enumerate}
        \item $v$ is an articulation point.
        \item The edges $(u,v)$ and $(v,w)$ are both bridges.
        \item For any vertex $v' \neq v$, $\card{d(u,v')-d(w,v')} = 2$.
    \end{enumerate}
\end{lemma}

\begin{proof}
    We first prove the first two conditions are equivalent. If $v$ is an articulation point, then $v$ is in two different blocks. So edge $(u,v)$ and $(v,w)$ are in different blocks, which means $v$ has degree $1$ in both blocks. So both blocks are trivial, which means $(u,v)$ and $(v,w)$ are both bridges. If $(u,v)$ and $(v,w)$ are both bridges, then deleting either $(u,v)$ or $(v,w)$ will disconnect $u$ and $w$, which means deleting $v$ will also disconnect $u$ and $w$.

    Next we prove that the third condition is equivalent to the first two. Suppose $v$ is an articulation point. Since $v$ has degree $2$, the graph $G \setminus \{v\}$ has only two components, one containing $u$ and the other containing $w$. For any vertex $v' \neq v$, without loss of generality, suppose $v'$ is in the same component as $u$ in $G \setminus \{v\}$. Since $(u,v)$ and $(v,w)$ are both bridges in $G$, any path between $v'$ and $w$ contains $u$ and $v$. So $d(v',w)=d(v',u)+2$. 

    If $v$ is not an articulation point, then $u$ and $w$ are connected in $G \setminus \{v\}$. Let $(u=v_0,v_1,v_2,\dots,v_k=w)$ be the shortest path between $u$ and $w$ in $G \setminus \{v\}$. For any vertex $v_i$ on the path, the distance between $v_i$ and $u$ (resp. $w$) in $G \setminus \{v\}$ is $i$ (resp. $k-i$). Consider the shortest path between $u$ and $v_i$ in $G$. If this path does not contain $v$, then it is the same as the path in $G \setminus \{v\}$. In this case, $d(u,v_i) = i$. If the shortest path contains $v$, then $v$ must be the second last vertex on the path and $w$ be the third last one. In this case, $d(u,v_i) = k-i+2$. So $d(u,v_i) = \min \{ i, k-i+2 \}$. Similarly, we also have $d(v_i,w) = \min \{i+2,k-i\}$. Let $v' = v_{\lfloor k/2 \rfloor}$. Since $\card{i - (k-i)} \le 1$, we have $i<k-i+2$ and $k-i<i+2$, which means $\card{d(u,v')-d(w,v')} = \card{i - (k-i)} \le 2$.
\end{proof}

Next, we prove that if $\alpha$ is small or $\beta$ is large, the cost of graphic TSP is bounded away from $n$. The following lemma shows that if the size of maximum matching of a graph is small, then the cost of the graphic TSP is large.

\begin{lemma} \label{lem:sma-mat}
    For any $\eps>0$, if the maximum matching of a graph $G$ has size at most $\frac{(1-\eps)n}{2}$, then the cost of graphic TSP of $G$ is at least $(1+\eps)n$.
\end{lemma}

\begin{proof}
    Suppose the optimal TSP tour is $(v_0,v_1,\dots,v_{n-1},v_n=v_0)$. Since the size of maximum matching in $G$ is at most $\frac{(1-\eps)n}{2}$, there are at most $\frac{(1-\eps)n}{2}$ edges between pairs $(v_i,v_{i+1})$ where $i$ is even (resp. odd). So there are at least $\eps n$ pairs of $(v_i,v_{i+1})$ that have distance at least $2$, which means that the optimal cost of TSP tour of $G$ is $\sum_{i=1}^{n-1} d(v_i,v_{i+1}) \ge n + \eps n = (1+\eps) n$.
\end{proof}

The following lemma shows that if $\beta$ is large, the cost of graphic TSP is large.

\begin{lemma} \label{lem:bad-bri}
    For any $\eps>0$, if a connected graph $G$ has $\eps n$ bad vertices, then the cost of graph-TSP of $G$ is at least $(1+\eps)n-2$.
\end{lemma}

\begin{proof}
    We first prove by induction on the number of vertices that a graph with $k$ bad vertices has $k-1$ bridges. The base case is when $n=2$, the graph has $k=2$ bad vertices and $1=k-1$ bridge.

    For the induction step, the graph has $n$ vertices with $n\ge 3$. If $G$ has no degree $1$ vertices, then the graph has $k$ articulation points with degree $2$. By Lemma~\ref{lem:deg2-art}, any edge incident on a degree 2 articulation point is a bridge. So each bad vertex is incident on $2$ bridges. On the other hand, a bridge is incident on at most $2$ vertices. So there are at least $\frac{2k}{2} = k$ bridges in $G$. Next, suppose $G$ has degree $1$ vertices. Let $v$ be an arbitrary such vertex and let $u$ be its neighbor. Since $G$ is connected and $n \ge 3$, $u$ must has degree at least $2$, since otherwise $u$ and $v$ are not connected to other vertices in $G$. Consider the graph $G \setminus \{v\}$, if $u$ is a bad vertex in $G$, $u$ has degree $1$ in $G \setminus \{v\}$ and is still a bad vertex. So the number of bad vertices in $G \setminus \{v\}$ is $k-1$. By induction hypothesis, $G \setminus \{v\}$ has at least $k-2$ bridges. $G$ has at least $k-1$ bridges since $(u,v)$ is also a bridge.

    So $G$ has at least $\eps n-1$ bridges, and the number of blocks in $G$ is at least $\eps n-1$. By Lemma~\ref{lem:blo-size}, the cost of graph-TSP of $G$ is at least $n+\eps n - 2 = (1+\eps)n -2$.
\end{proof}

Finally, the following lemma shows that the cost of graphic TSP is at most $(2-\frac{2}{7}(\hat{\alpha}-2\beta))n$.

\begin{lemma} \label{lem:pass}
    If a graph has a matching $M$ of size $\alpha' n$ and the graph has $\beta n$ bad vertices, the cost of graphic TSP of $G$ is at most $(2-\frac{2}{7}(\alpha'-2\beta))n$.
\end{lemma}

\begin{proof}
    Let $G_1,G_2,\dots,G_k$ be the block decomposition of $G$. Let $n_i$ be the size of $G_i$. If $\card{n_i} \ge 3$, by Lemma~\ref{lem:bi-mat}, the cost of the graphic TSP of $G_i$ is at most $2n_i-3$ since any non-empty graph has a matching of size at least $1$. If $\card{n_i} = 2$, then the graphic TSP of $G_i$ is exactly $2 = 2n_i -2$. Suppose $G$ has $\ell$ non-trivial blocks. Then by Lemma~\ref{lem:blo-tsp} the cost of graphic TSP of $G$ is at most $\sum_{i=1}^k(2n_i-2) - \ell$, which equals to $2n - 2 - \ell$ by Lemma~\ref{lem:blo-num}. 

    Let $m_i$ be the size of maximum matching in $G_i$ if $G_i$ is a non-trivial block, and let $m_i=0$ if $G_i$ is a trivial block. By Lemma~\ref{lem:bi-mat}, the cost of the graphic TSP of $G_i$ is at most $2n_i-2-\frac{2m_i}{3}$. For any non-trivial block $G_i$, $M \cap G_i$ is a matching in $G_i$. So the size of maximum matching in $G_i$ is at least the number of edges in $M \cap G_i$. So by Lemma~\ref{lem:blo-tsp} and Lemma~\ref{lem:blo-num}, the cost of graphic TSP of $G$ is at most $\sum_{i=1}^k (2n_i-2-\frac{2}{3}m_i) = 2n-2 -\frac{2}{3} \card{M'}$, where $M'$ is the set of edges in $M$ that are not bridges in $G$. Let $B$ be the number of bridges in $G$. We have $2n-2-\frac{2}{3}\card{M'} \le 2n-2-\frac{2}{3} (\card{M}-B)$.

    So there are two upper bounds of the graphic TSP of $G$ --- $2n-2-\ell$ and $2n-2-\frac{2}{3} (\card{M}-B)$. Which bound is better depends on the number of bridges $B$.
    
    If $B \le (\frac{4}{7}\alpha'+\frac{6}{7}\beta)n$, the cost of graphic TSP of $G$ is at most 
    \begin{align*}
        2n-2-\frac{2}{3} (\card{M}-B) \le 2n-\frac{2}{3} (\frac{3}{7}\alpha'-\frac{6}{7}\beta)n = (2-\frac{2}{7}(\alpha'-2\beta))n
    \end{align*}

    If $B > (\frac{4}{7}\alpha'+\frac{6}{7}\beta)n$, consider the bridge-block tree $T$ of $G$. $T$ has at least $B$ edges and at least $B+1$ vertices. Since $T$ is a tree, there are at least $\frac{B}{2}$ vertices of degree at most $2$. For any vertex $v_T$ of degree at most $2$ in $T$, if the vertex $v_T$ represents a single vertex $v$ in $G$, then $v$ is either a degree $1$ vertex or a degree $2$ articulation point in $G$, otherwise $v_T$ represents a 2-edge-connected component of size at least 2 in $G$. So There are at least $\frac{B}{2}-\beta n \ge (\frac{2}{7}\alpha' - \frac{4}{7}\beta)n$ 2-edge-connected components of size at least 2. Since any 2-edge-connected component of size at least 2 has no bridge, each such component of $G$ contains at least 1 non-trivial block in $G$, implying that $\ell \ge \frac{2}{7}(\alpha'-2\beta)n$. So the cost of graphic TSP of $G$ is at most $2n-2-\ell \le (2-\frac{2}{7}(\alpha'-2\beta))n$.
\end{proof}

We summarize the ideas in this section and prove the following lemma.

\begin{lemma} \label{lem:graph-main}
    For any $c_0>1$ and $\delta>0$, suppose $\hat{\alpha}\le \alpha \le c_0 \hat{\alpha} + \delta$ and $\hat{\beta} - \delta \le \beta \le \hat{\beta}$. Then $(2-\frac{2}{7}(\hat{\alpha}-2\hat{\beta}))n$ is an approximation of the size of graphic TSP within a factor of $2-\frac{1}{7c_0}+\delta$.
\end{lemma}

\begin{proof}
    Let $\hat{T} = (2-\frac{2}{7}(\hat{\alpha}-2\hat{\beta}))n$. Since $\hat{\beta} \ge \beta$ and $\hat{\alpha} \le \alpha$, by Lemma~\ref{lem:pass}, $T \le \hat{T}$.

    Then we prove that $\hat{T} \le (2-\frac{1}{7c_0}+\delta) T$. By Lemma~\ref{lem:sma-mat} and Lemma~\ref{lem:bad-bri}, $T \ge \max\{(2-2\alpha)n,(1+\beta)n-2\}$, which means 
    \begin{equation*}
        (2-\frac{1}{7c_0}+\delta) T \ge (2-\frac{1}{7c_0}) \max\{(2-2\alpha)n,(1+\beta)n\} - 4 + \delta n 
    \end{equation*}
    On the other hand, $\hat{T} \le (2-\frac{2}{7}(\frac{\alpha}{c_0}-2\beta))n + \frac{6}{7} \delta n $ since $c_0\hat{\alpha}+\delta \le \alpha$ and $\hat{\beta} \le \beta+\delta$. For sufficient large $n$, we have $\delta n - 4 \ge \frac{6}{7}\delta n$, so it is sufficient to prove that $\frac{2-\frac{2}{7}(\frac{\alpha}{c_0}-2\beta)}{\max\{2-2\alpha,1+\beta\}} \le 2-\frac{1}{7c_0}$ for any $0 \le \alpha , \beta \le 1$ and $c_0 \ge 1$.

    Let $\gamma = \frac{\alpha}{c_0} - 2\beta$, $1+\beta = 1 + (\frac{\alpha}{c_0}-\gamma)/2$, so if we fix $\gamma$, $\max\{2-2\alpha,1+\beta\}$ is minimized when $2-2\alpha = 1+(\frac{\alpha}{c_0} - \gamma)/2$. In this case $\alpha = \frac{(2+\gamma)c_0}{4c_0+1}$ and $\max\{2-2\alpha,1+\beta\} = \frac{4c_0+2}{4c_0+1} - \frac{2c_0}{4c_0+1} \gamma$. If $\gamma \le \frac{1}{2c_0}$, 
    \begin{align*}
        \frac{2-\frac{2}{7}(\alpha-2\beta)}{\max\{2-2\alpha,1+\beta\}} &\le \frac{2-\frac{2}{7} \gamma}{\frac{4c_0+2}{4c_0+1} - \frac{2c_0}{4c_0+1} \gamma} = \frac{4c_0+1}{7c_0} - \frac{2-\frac{4c_0+2}{7c_0}}{\frac{4c_0+2}{4c_0+1} - \frac{2c_0}{4c_0+1} \gamma} \\ 
        & \le \frac{4c_0+1}{7c_0} + 2 - \frac{4c_0+2}{7c_0} = 2 - \frac{1}{7c_0}
    \end{align*}
    If $\gamma > \frac{1}{2c_0}$, $\frac{2-\frac{2}{7}(\alpha-2\beta)}{\max\{2-2\alpha,1+\beta\}} < \frac{2-\frac{1}{7c_0}}{1} = 2 - \frac{1}{7c_0}$ since $\beta \ge 0$. So $\hat{T} \le (2-\frac{1}{7c_0} + \delta) T$.
\end{proof}

By Lemma~\ref{lem:graph-main}, we immediately have the following theorem.

\begin{theorem} \label{thm:graph-main}
    For any $\delta>0$ and $c_0 \ge 1$. Given a graph $G$ with maximum matching size $\alpha n$, suppose there is an algorithm that uses pair queries, runs in $t$ time, and with probability at least $2/3$, outputs an estimate of the maximum matching size $\hat{\alpha} n$ such that $\hat{\alpha} \le \alpha \le c_0\hat{\alpha}+\delta$. Then there is an algorithm that approximates the cost of graphic TSP of $G$ to within a factor of $2-\frac{1}{7c_0}+\delta$, using distance queries, in $t+\tilde{O}(n/\delta^2)$ time with probability at least $3/5$.
\end{theorem}

\begin{proof}
    We first use the algorithm in the assumption to obtain an estimate $\hat{\alpha}n$ of the size of maximum matching $\alpha n$. The following analysis is based on the event that this algorithm is run successfully, which has probability $2/3$. 
    
    We then sample $N=\frac{100}{\delta^2}$ vertices. For each sampled vertex $v$, we first query the distance between $v$ and every vertex in $G$ to obtain the degree of $v$. If $v$ has degree $2$, suppose $u$ and $w$ are the neighbors of $v$. We query the distance from $u$ and $w$  to every vertex in $G$. By Lemma~\ref{lem:deg2-art}, $v$ is an articulation point if and only if there is no vertex $v'$ such that $\card{d(u,v')-d(w,v')} \le 1$. So we can check if $v$ is a bad vertex with $O(n)$ distance queries and time. Suppose there are $\beta n$ bad vertices in $G$ and $(\hat{\beta}-\delta/2)N$ sampled vertices are bad. By Chernoff bound, the probability that $\card{\beta-\hat{\beta}+\delta/2} > \delta/2$ is at most $2 e^{\frac{\delta^2 N^2}{16}} < 1/15$. We analyze the performance based on the event that $\beta\le \hat{\beta} \le \beta + \delta$. 

    By Lemma~\ref{lem:graph-main}, $(2-\frac{2}{7}(\hat{\alpha}-2\hat{\beta}))$ is a $(2-\frac{1}{7c_0}+\delta)$ approximation of the size of graphic TSP of $G$. The probability of failure is at most $1/3+1/15 = 2/5$.
\end{proof}

\smallskip
\noindent
{\bf Proof of Theorem~~\ref{thm:graphic_tsp_1.5}:}
The following theorem whose proof appears in Appendix~\ref{sec:app-mat}, gives an algorithm for matching size estimation that only uses {\em pair queries} -- given a pair of vertices, is there an edge between them? Note that any pair query can be simulated by a single query to the distance matrix in a graphic TSP instance. 

\begin{theorem} \label{lem:max-mat}
   For any $\eps > 0$,  there is an algorithm that uses pair queries, runs in $\tilde{O}(n^{1.5}/\eps^2)$ time, and with probability $2/3$, outputs an estimate of the size of a maximal matching within an additive error $\eps n$. 
\end{theorem}

Substituting the above result in Theorem~\ref{thm:graph-main} and using the fact that a maximum matching has size at most twice the size of a maximal matching (setting $c_0 = 2$, and $\delta = \eps$), we obtain Theorem~\ref{thm:graphic_tsp_1.5}.

\smallskip
\noindent
{\bf Proof of Theorem~~\ref{thm:graphic_tsp_linear}:}
Kapralov et al. \cite{kapralov2020space} give an algorithm that uses $\tilde{O}(d)$ queries (also $\tilde{O}(d)$ time) to approximate the size of maximum matching in a graph with average degree $d$ in the neighbor query model (the approximation ratio is a very large constant). Together with a reduction in \cite{onak2012near}, this implies a pair query algorithm that uses $\tilde{O}(n)$ queries and time to estimate matching size to a constant factor. 
Combined with Theorem~\ref{thm:graph-main}, this implies Theorem~\ref{thm:graphic_tsp_linear}.

\subsection{An $O(n)$ Space $(\frac{11}{6})$-Approximate Streaming Algorithm for Graphic TSP}

We show here that our approach for obtaining a sublinear-time algorithm for graphic TSP can be extended to the insertion-only streaming model to obtain for any $\eps > 0$, an $(\frac{11}{6}+\eps)$-approximate estimate of the graphic TSP cost using $O(n/\eps^2)$ space, proving Theorem~\ref{thm:graphic_tsp_stream_main}. In the streaming model, we assume that the input to graphic TSP is presented as a sequence of edges of the underlying graph $G$. Any algorithm for this model, clearly also works if instead the entries of the distance matrix are presented in the stream instead -- an entry that is $1$ corresponds to an edge of $G$, and it can be ignored otherwise as a non-edge. 

Given a stream containing edges of a graph $G(V,E)$, our algorithm performs the following two tasks in parallel:

\begin{itemize}
    \item Find a maximal matching $M$ in $G$ -- let $\alpha n$ denote its size.
    \item Estimate the number of bridges in the maximal matching $M$, say $\beta n$, to within an additive error of $\eps n$.
\end{itemize}

The algorithm outputs $(2-\frac{2}{3}(\alpha-\beta))n$ as the estimated cost of graphic TSP of $G$. 

In an insertion-only stream, it is easy to compute a maximal matching $M$ using $O(n)$ space: we start with $M$ initialized to an empty set, and add a new edge $(u,v)$ into the matching $M$ iff neither $u$ nor $v$ are already in $M$. It is also easy to check if an edge $e$ is a bridge in insertion-only stream with $O(n)$ space. We can do this by maintaining a disjoint-set data structure. Whenever an edge arrives (other than $e$), we merge the connected components of its endpoints. If there is only one component remaining at the end of the stream, then $e$ is not a bridge, and otherwise, $e$ is a bridge. 

To estimate the number of bridges in the maximal matching, we sample $N=100/\eps^2$ edges in the matching, and run in parallel $N$ tests where each test determines whether or not the sampled edge is a bridge.  We use $O(n/\eps^2)$ space in total since we sample $N=O(1/\eps^2)$ edges. Suppose there are $\bar{\beta}$ sampled edges are bridges, then by Chernoff bound, $\hat{\beta} n = \frac{\bar{\beta} \card{M}}{N}$ is an approximation of $\beta n$ to within additive error $\eps n$ with probability at least $9/10$.

As stated, this gives us a two-pass algorithm: the first pass for computing the matching $M$, and the second pass for estimating the number of bridges in $M$. However, we can do both these tasks in parallel in a single pass as follows: at the beginning of the stream, we start the process of finding connected components of graph $G$. Whenever an edge $e$ is added to $M$, if $\card{M} < N$, then we create a new instance $I_e$ of the connectivity problem that ignores the edge $e$. This clearly allows us to test whether or not $e$ is a bridge. Once $\card{M}>N$, then whenever an edge $e$ is added to $M$, with probability $\frac{N}{\card{M}}$, we drop uniformly at random an existing instance, say $I_{e'}$ of connectivity, and create a new instance $I_e$ of connectivity that only ignores edge $e$ (we insert back the edge $e'$ into $I_{e}$). Since there are at most $N$ instances of connectivity that are running in parallel, the algorithm uses $O(nN) = O(n/\eps^2)$ space.

We now prove that the algorithm gives a good approximation of the cost of graphic TSP.

\begin{lemma} \label{lem:stream-TSP}
    If a graph $G$ has a maximal matching $M$ of size $\alpha n$, and there are $\beta n$ edges in $M$ that are bridges in $G$, then the cost of graphic TSP in $G$ is at most $(2-\frac{2}{3}(\alpha-\beta))n$, and at least $\frac{6}{11}(2-\frac{2}{3}(\alpha-\beta))n$.
\end{lemma}

\begin{proof}
    Since there are at least $(\alpha-\beta)n$ edges in the matching $M$ that are not a bridge, by Lemma~\ref{lem:blo-tsp} and Lemma~\ref{lem:bi-mat}, the cost of graphic TSP of $G$ is at most $(2-\frac{2}{3}(\alpha-\beta))n$.
    
    On the other hand, since $M$ is a maximal matching of $G$, the size of maximum matching of $G$ is at most $2\alpha n$. By Lemma~\ref{lem:sma-mat}, the cost of graphic TSP is at least $(2-4\alpha)n$. Graph $G$ also contains at least $\beta n$ bridges, so by Lemma~\ref{lem:blo-size}, the cost of graphic TSP is also at least $(1+\beta)n$. 

    To prove the lemma, it is sufficient to prove that for any $0 \le \beta \le \alpha \le 1$, we have $2-\frac{2}{3}(\alpha-\beta) \le \frac{11}{6} \max \{1+\beta,2-4\alpha\}$. Let $\gamma = \alpha - \beta$. $1+\beta = 1 + \alpha - \gamma$. So $\max\{1+\beta,2-4\alpha\} \ge 2-4(\frac{1}{5}(1+\gamma)) = \frac{6}{5}-\frac{4}{5}\gamma$. If $\gamma \le \frac{1}{4}$, $\frac{2-\frac{2}{3}(\alpha-\beta)}{\max\{1+\beta,2-4\alpha\}} \le \frac{2-\frac{2}{3}\gamma}{\frac{6}{5}-\frac{4}{5}\gamma} = \frac{5}{6} + \frac{5}{6-4\gamma} = \frac{11}{6}$. If $\gamma > \frac{1}{4}$, $2-\frac{2}{3}(\alpha-\beta) < \frac{11}{6}$, while $\max\{1+\beta,2-4\alpha\} \ge 1$ since $\beta>0$.
\end{proof}

By Lemma~\ref{lem:stream-TSP}, the expression $(2-\frac{2}{3}(\alpha-\beta))n$ gives us an $11/6$-approximate estimate to the cost of graphic TSP of $G$. Since we can exactly compute $\alpha$ and approximate $\beta$ with additive error $\eps$ in a single-pass streaming algorithm that uses $O(n/\eps^2)$ space, we have the following theorem:

\begin{theorem} \label{thm:stream-TSP}
    For any $\eps > 0$, there is a single-pass randomized streaming algorithm that estimates the cost of graphic TSP of $G$ to within a factor of $(\frac{11}{6}+\eps)$, in an insertion-only stream, using $O(n/\eps^2)$ space with probability at least $9/10$.
\end{theorem}

\subsection{Extension to the Massively Parallel Computing Model}

In the massive parallel computing (MPC) model, the input graph $G$ is partitioned across multiple machines which are able to communicate with one another, and the memeory allocated to each machine is sublinear in the total input size. The computation proceeds in synchronous rounds where in any round, each machine runs a local algorithm on the data assigned to the machine. No communication between machines is allowed during a round. Between the rounds, machines can communicate with each other so long as each machine sends or receives a communication no more than its memory. Any data output from a machine must be computed locally from the data residing on the machine and initially the input data is distributed across machines in an arbitrary manner. The goal is to minimize the total number of rounds. 

We extend our algorithms for the query model and the streaming model to the MPC model. In both query model and streaming model, we approximate the size of maximal matching and then approximate the number of articulation points or bridges in the graph to get the upper bounds and lower bounds of the cost of graphic TSP of the graph. The difference is that in streaming model, we can also easily compute a solution to a maximal matching (and not just estimate its size) in contrast to the query model, which results in a better approximating ratio. In general, however, the task of finding an approximate matching can be much harder than approximating the size of the maximum matching. 

There have been many works studying the connectivity problem and matching problem in MPC model. Since there is a trade-off between the size of memory and the number of rounds, there are many different ``state-of-the-art'' results depending on the size of the memory in each machine. So rather than give algorithms for specific tradeoffs, we give two general results that translate various algorithms for the connectivity problem and matching problem to an algorithm for estimating the cost of the graphic TSP problem.

The following two corollaries follow from the proof of Lemma~\ref{lem:graph-main} and Theorem~\ref{thm:stream-TSP}. If a graph $G$ has maximum size $\card{M^{\star}}$, we say a number $\card{M}$ is a {\em $(\alpha,\eps n)$-approximation} of $\card{M^{\star}}$ if $\card{M} \le \card{M^{\star}} \le \alpha \card{M}+\eps n$.

\begin{corollary} \label{cor:map}
    If there is an algorithm that computes an $(\alpha,\eps n)$-estimation of the size of maximum matching in MPC model that uses $O(f_1(n,\eps))$ rounds, where each machine has $O(g_1(n,\eps))$ space with probability at least $9/10$, and there is an algorithm that checks if a graph is connected in MPC model that uses $O(f_2(n,\eps))$ rounds, where each machine has $O(g_2(n,\eps))$ space with probability at least $1-\eps^2/500$. Then there is an algorithm that approximates the size of graphic TSP within a factor of $(\frac{14\alpha-1}{7\alpha}+\eps)$ in $O(f_1(n,\eps)+f_2(n,\eps))$ rounds, where each machine has $\max\{O(g_1(n,\eps),O(g_2(n,\eps)/\eps^2) \}$ space with probability at least $2/3$.
\end{corollary}

\begin{proof}[Proof Sketch]
    We first run the MPC algorithm that estimates the matching size. Then sample $\frac{100}{\eps^2}$ vertices, and check if any of them are bad vertices so as to estimate the total number of bad vertices in the graph. To check if a vertex is a bad vertex, we first check if it has degree 1 or 2, then check if it is a articulation point by checking the connectedness of the graph when we delete the vertex and all edges incident on it. We can test all sampled vertices simultaneously if each machine has $\Omega(g_2(n,\eps)/\eps^2)$ space. The correctness follows from the same argument as Lemma~\ref{lem:graph-main}. The failure probability of the matching algorithm is at most $1/10$ and the failure probability that we make a mistake on at least one sampled vertex is at most $\frac{100}{\eps^2} \cdot \frac{\eps^2}{500} = 1/5$. So the total probability of failure is at most $1/3$, giving as the desired result.
\end{proof}

\begin{corollary}
    Suppose there exists an algorithm that computes an $(\alpha,\eps n)$-approximation of maximum matching in MPC model using $O(f_1(n,\eps))$ rounds, where each machine has $O(g_1(n))$ space with probability at least $9/10$, and there is an algorithm that checks if a graph is connected in MPC model that uses $O(f_2(n,\eps))$ rounds, where each machine has $O(g_2(n,\eps))$ space with probability at least $1-\eps^2/500$. Then there is an algorithm that approximates the cost of graphic TSP to within a factor of $(\frac{6\alpha-1}{3\alpha}+\eps)$ in $O(f_1(n,\eps)+f_2(n))$ rounds, where each machine has $\max\{O(g_1(n,\eps),O(g_2(n,\eps)/\eps^2) \}$ space with probability at least $2/3$.
\end{corollary}
\begin{proof}[Proof Sketch]
    The proof is similar to the proof of Corollary~\ref{cor:map}. The difference is that now we can find an approximate matching instead of just estimating the matching size. So we can now sample $\frac{100}{\eps^2}$ edges in the approximate matching and estimate the number of bridges in the matching. The correctness follows from a similar argument as in the proof of Theorem~\ref{thm:stream-TSP}.
\end{proof}

\section{$(1.625)$-Approximation for $(1,2)$-TSP Cost in $\tilde{O}(n^{1.5})$ Time}
\label{sec:1,2_TSP}

In this section, we give an algorithm that for any $\delta > 0$, approximates the cost of the minimum $(1,2)$-TSP to within a factor of $1.625+\delta$ with $\tilde{O}(n^{1.5}/\delta^2)$ queries. The idea of the algorithm is to approximate the size of a maximal ``matching pair'' of $G$. In a graph $G$, a \textit{matching pair} $(M_1,M_2)$ is a pair of edge-disjoint matchings. A \textit{maximal} matching pair is a matching pair $(M_1,M_2)$ such that for any edge $e \not \in M_1 \cup M_2$, neither $M_1 \cup \{e\}$ nor $M_2 \cup \{e\}$ is a matching. 
The size of a matching pair $(M_1,M_2)$ is the sum of the sizes of $M_1$ and $M_2$. The following lemma shows that the size of any maximal matching pair is lower bounded by the size of maximum matching in the graph. 

\begin{lemma} \label{lem:mat-pair-size}
    Suppose $M$ is a matching in a graph $G$. Then any maximal matching pair $(M_1,M_2)$ in $G$ has size at least $\card{M}$.
\end{lemma}

\begin{proof}
    Let $X_1$ be the set of vertices matched in both $M$ and $M_1$, and $X_2$ be the set of vertices matched in both $M$ and $M_2$. We have $\card{X_1}+\card{X_2} \le 2\card{M_1}+2\card{M_2}$ since $M_1$ and $M_2$ are both matchings. On the other hand, for any edge $e \in M$, if $e$ is either in $M_1$ or $M_2$, then both of its endpoints are in $X_1$ or $X_2$. If $e$ is neither in $M_1$ or $M_2$, then there are edges $e_1\in M_1$ and $e_2 \in M_2$ that share an endpoint with $e$ since $(M_1,M_2)$ is a maximal matching pair. So both $X_1$ and $X_2$ contain at least one endpoint of $e$. In both case $e$'s  endpoints appear twice in $X_1$ and $X_2$. So $\card{X_1}+\card{X_2} \ge 2\card{M}$, which means $\card{M_1}+\card{M_2} \ge \card{M}$.
\end{proof}

We next show that if a graph has a matching pair of large size, then the cost of $(1,2)$-TSP is not very large. 

\begin{lemma} \label{lem:mat-tsp-ot}
    If a graph $G$ with $n$ vertices contains a matching pair $(M_1,M_2)$ of size $X$, then the cost of $(1,2)$-TSP of $G$ is at most $2n-\frac{3}{4}X$.
\end{lemma}

\begin{proof}
    Since $M_1$ and $M_2$ are both matchings, $M_1 \cup M_2$ only contains paths and  cycles of even length. We delete one edge from each cycle in $M_1 \cup M_2$, resulting in a graph that only contains paths. Since the cycles in $M_1 \cup M_2$ are of even length, the size of any cycle is at least $4$. We deleted at most $\frac{1}{4} X$ edges, so $G$ contains a set of vertex disjoint paths (including some of length 0,  corresponding to isolated vertices), with total size at least $\frac{3}{4}X$. Construct a TSP tour by ordering the paths arbitrarily, orienting each one, and connecting the end of one path with the start of the next, cyclically. The tour contains at least $\frac{3}{4}X$ edges of weight $1$, while the remaining edges are of weight 2. So the cost of the tour is at most $2n-\frac{3}{4}X$.
\end{proof}

By Lemma~\ref{lem:mat-pair-size}, the maximum matching size is upper bounded by the size of any maximal matching pair. It follows that  if the maximum matching size is small, the cost of $(1,2)$-TSP is large. 

\begin{lemma} \label{lem:sma-mat-ot}
    For any $\eps>0$, if the maximum matching of a graph $G$ has size at most $\frac{(1-\eps)n}{2}$, then the cost of $(1,2)$-TSP of $G$ is at least $(1+\eps)n$.
\end{lemma}

The proof of Lemma~\ref{lem:sma-mat-ot} is similar to the proof of Lemma~\ref{lem:sma-mat} and we omit it here. By Lemma~\ref{lem:mat-tsp-ot} and Lemma~\ref{lem:sma-mat-ot}, if we can approximate the size of an arbitrary maximal matching pair, then we will get a good approximation of the cost of the $(1,2)$-TSP.

\begin{theorem} \label{lem:mat-pair}
    There is an algorithm that uses pair queries, with probability at least $2/3$, approximates the size of a maximal matching pair to within an additive error of $\eps n$ in $\tilde{O}(n^{1.5}/\eps^2)$ time. 
\end{theorem}

The algorithm in Theorem~\ref{lem:mat-pair} is given in Appendix~\ref{sec:mat-pair}. With Theorem~\ref{lem:mat-pair}, we can approximate the cost of $(1,2)$-TSP in a graph $G$ by the size of a maximal matching pair.

\begin{theorem}
    For any $\delta>0$, there is an algorithm that with probability at least $2/3$ estimates the optimal cost of a $(1,2)$-TSP instance to within a factor of $(1.625+\delta)$ using $\tilde{O}(n^{1.5}/\delta^2)$ queries.
\end{theorem}

\begin{proof}
    Let $\eps = \delta/2$. We use the algorithm in Theorem~\ref{lem:mat-pair} that approximates the size of a maximal matching pair. Suppose the output of the algorithm is $\bar{X}$. Then,  by Theorem~\ref{lem:mat-pair}, there is a maximal matching pair of size $X$ such that $\card{X-\bar{X}} \le \eps n$ . We output the cost of the $(1,2)$-TSP of $G$ to be $\bar{T} = 2n-\frac{3}{4}(\bar{X}-\eps n)$. Suppose the optimal $(1,2)$-TSP has cost $T$. By Lemma~\ref{lem:mat-tsp-ot}, $T \le 2n - \frac{3}{4}X \le 2n - \frac{3}{4}(\bar{X}-\eps n) = \bar{T}$. On the other hand, by Lemma~\ref{lem:sma-mat-ot}, the size of maximum matching in $G$ is at least $(2n-T)/2$. So by Lemma~\ref{lem:mat-pair-size}, $X \ge (2n-T)/2$, which means $\bar{X} \ge (2n-T)/2 - \eps n$. So $\bar{T} \le 2n- \frac{3}{4}(n-T/2-2 \eps n) < 1.25n+0.375T+\delta n$. Since $T$ is the cost of $(1,2)$-TSP of $G$, which is at least $n$, we have $\bar{T} \le (1.625+\delta) T$.
\end{proof}

\begin{remark}
    The algorithm can be generalized to insertion-only streaming model. In insertion-only streaming model, we can compute a maximal matching pair as follows: we set $M_1$ and $M_2$ as empty set before the stream. Whenever an edge $e$ comes, we first check if there is an edge in $M_1$ that shares an endpoint with $e$. If not, then we add $e$ into $M_1$. Otherwise, we check if there is an edge in $M_2$ that shares and endpoint with $e$. If not, then we add $e$ into $M_2$. So we get an algorithm that only uses $O(n)$ space to compute a maximal matching pair. We have the following corollary.
\end{remark}

\begin{corollary}
    There is an insertion-only streaming algorithm that estimates the cost of $(1,2)$-TSP of a graph $G$ within a factor of $1.625$ using $O(n)$ space. 
\end{corollary}

\section{An $\Omega(n^2)$ Query Lower Bound for Approximation Schemes}

\label{sec:ptas_lower_bound}

\iffalse
In this section, we prove that there exists an $\eps_0>0$, such that any randomized query algorithm for graphic or $(1,2)$-TSP that returns a $(1+\eps_0)$-approximate estimate of TSP cost, requires $\Omega(n^2)$ queries. In order to prove this, we will first design a new query model for the $\Equ$ problem, and prove an $\Omega(n^2)$ query lower bound for $\Equ$ in this new query model (where $n$ is the number of variables) and then use a standard reduction from $\Equ$ to 3SAT to also obtain an $\Omega(n^2)$ query lower bound for 3SAT. Finally, we use the reduction from 3SAT to $(1,2)$-TSP in \cite{papadimitriou1993traveling} to prove the lower bound for $(1,2)$-TSP; with some additional changes sketched at the end of the section, we also get an identical lower bound for graphic TSP.

The idea of proving lower bound for APX-hard problem by reduction from 3SAT is similar to the idea used in \cite{bogdanov2002lower}.  However, in \cite{bogdanov2002lower}, the authors study lower bounds for problems in sparse graphs and hence the query model uses only neighbor queries. So in their query model, the lower bound for 3SAT and $\Equ$ are $\Omega(n)$. In order to prove an $\Omega(n^2)$ query lower bound in the pair query model, we need to design a new query model for 3SAT and $\Equ$.  As in \cite{bogdanov2002lower}, we find it convenient to start with the lower bound for $\Equ$, derive a lower bound for 3SAT, and then for graphic and $(1,2)$-TSP.
\fi

In this section, we prove that there exists an $\eps_0>0$, such that any query algorithm for graphic or $(1,2)$-TSP that returns a $(1+\eps_0)$-approximate estimate of optimal cost, requires $\Omega(n^2)$ queries. 
In order to prove this, we design a new query model for the 3SAT problem and show an $\Omega(n^2)$ query lower bound for 3SAT in this model. We then use a reduction from 3SAT to $(1,2)$-TSP in \cite{papadimitriou1993traveling} to prove the lower bound for $(1,2)$-TSP; with some additional changes, we also get an identical lower bound for graphic TSP.

The idea of proving query lower bound for APX-hard problems by reduction from 3SAT is similar to the idea used in \cite{bogdanov2002lower}, and we follow their general approach.  However, in \cite{bogdanov2002lower}, the authors study lower bounds for problems in sparse graphs and hence the query model uses only neighbor queries. So in their query model, the lower bound for 3SAT is $\Omega(n)$. In order to prove an $\Omega(n^2)$ query lower bound in the pair query model, we need to design a new query model for 3SAT.  

%\subsection{An $\Omega(n^2)$ Query Lower Bound for 3SAT Problem}

In the 3SAT problem, we are given a 3CNF instance on $n$ variables, and the goal is to estimate the largest fraction of clauses that can be satisfied by any assignment. The algorithm is allowed to perform only one kind of query: is a variable $x$ present in a clause $c$? If the answer is yes, then the algorithm is given the full information about all variables that appear in the clause $c$. The proof of the next theorem is deferred to Section~\ref{sec:app-sat}.

\begin{theorem} \label{thm:3SAT}
  For any $\eps>0$, any algorithm that with probability at least $2/3$ distinguishes between satisfiable 3CNF instances and 3CNF instances where at most $(7/8+\eps)$ fraction of clauses can be satisfied, needs $\Omega(n^2)$ queries.
\end{theorem}

\subsection{Reduction from 3SAT to $(1,2)$-TSP}
\label{sec:lb-TSP}

We will utilize an additional property of the hard instances of 3SAT in Theorem~\ref{thm:3SAT}, namely, each variable occurs the same constant number of times where the constant only depends on $\eps$. We denote the number of variables by $n$,  the number of clauses by $m$, and the number of occurrences of each variable by $k$; thus $m = kn/3$.

We use the reduction in \cite{papadimitriou1993traveling} to reduce a 3SAT instance to a $(1,2)$-TSP instance. In this reduction, there is a gadget for each variable and for each clause. Each of these gadgets has size at most $L = \Theta(k^2)$. Thus the $(1,2)$-TSP contains $N$ vertices where $N \le L (n+m) = \frac{L(k+3)n}{3}$. Let $G_{x_j}$ be the gadget of variable $x_j$ and $G_{c_i}$ be the gadget of clause $c_i$. There is a ground graph which is the same for each 3SAT instance. Each variable gadget is connected with the gadgets for clauses that contain that variable. The reduction satisfies the following property. If the 3SAT instance is satisfiable, then the $(1,2)$-TSP instance contains a Hamilton cycle supported only on the weight $1$ edges. On the other hand, if at most $m-\ell$ clauses can be satisfied in the 3SAT instance, the $(1,2)$-TSP cost is at least $N+\cei{\ell/2}$. Thus there is a constant factor separation between the optimal $(1,2)$-TSP cost in the two cases. However, what remains to be shown is that any query algorithm for $(1,2)$-TSP can also be directly simulated on the underlying 3SAT instance with a similar number of queries. 
The theorem below now follows by establishing this simulation.

\begin{theorem} \label{thm:lb-TSP}
  There is a constant $\eps_0$ such that any algorithm that approximates the $(1,2)$-TSP cost to within a factor of $(1+\eps_0)$ needs $\Omega(n^2)$ queries.
\end{theorem}

\begin{proof}
We consider the following stronger queries for $(1,2)$-TSP: for any query $(u,v)$, if $u$ is in a vertex gadget $G_{x_j}$ and $v$ is in a clause gadget $G_{c_i}$ (or vice versa) and $x_j$ occurs in $c_i$ in the 3SAT instance, then the algorithm is given all the edges incident on $G_{c_i}$. Otherwise the algorithm just learns if the there is an edge between $u$ and $v$.

Let $\eps = 1/16$, and let the values of $k$, $L$ and $N$ correspond to this choice for $\eps$ according to the redution in Section~\ref{sec:lb-TSP}. Let $\eps_0 = \frac{k}{32(k+3)L}$. Consider the $(1,2)$-TSP instance reduced from the 3SAT instance generated by the hard distribution in Theorem~\ref{thm:3SAT} with $\eps = 1/16$. If the 3SAT instance is perfectly satisfiable, then the $(1,2)$-TSP instance has a Hamilton cycle of cost $N$. If the 3SAT instance satisfies at most $(15/16)$-fraction of clauses, then each Hamilton cycle in the $(1,2)$-TSP instance has cost at least 
$$N+(1/8-\eps)m/2 = N+(1/8-\eps)kn/6 \ge (1+\frac{(1/8-\eps)k}{2(k+3)L})N = (1+\eps_0) N$$

For any query $(u,v)$ in the $(1,2)$-TSP instance, we can simulate it by at most one query in the corresponding 3SAT instance as follows: if $u$ is in a vertex gadget $G_{x_j}$ and $v$ is in a clause gadget $G_{c_i}$ (or vice versa), then we make a query of $x_j$ and $c_i$ in the 3SAT instance. If the 3SAT query returns YES and the full information of $c_i$, then we return all the edges incident on $G_{c_i}$ according to the reduction rule and the full information of $c_i$. If the 3SAT query returns NO or $(u,v)$ are not in a vertex gadget and a clause gadget respectively, we return YES if $(u,v)$ is an edge in the ground graph and NO otherwise.

By Theorem~\ref{thm:3SAT}, any algorithm that distinguishes a perfectly satisfiable 3SAT instance from an instance where at most $(15/16)$-fraction of the clauses can be satisfied needs $\Omega(n^2)$ queries. So any algorithm that distinguishes a $(1,2)$-TSP instance containing a Hamilton cycle of length $N$ from an instance that has minimum Hamilton cycle of cost $(1+\eps_0)N$ needs $\Omega(n^2)$ queries.

\end{proof}

\subsection{$\Omega(n^2)$ Lower Bound for Graphic TSP} 
\label{sec:lb_graphic_TSP}

We can reduce an instance of $(1,2)$-TSP to an instance of graphic TSP by adding a new vertex that is adjacent to all other vertices. By doing so, any pair of vertices in the new graph has a distance at most $2$. On the other hand, the cost of graphic TSP in the new graph differs by at most $1$ from the cost of $(1,2)$-TSP in the old graph. So the
$\Omega(n^2)$ query lower bound for $(1,2)$-TSP also holds for the graphic TSP problem.

\subsection{An $\Omega(n^2)$ Query Lower Bound for the 3SAT Problem} \label{sec:app-sat}
We first prove a lower bound of $\Equ$ problem. $\Equ$ is the problem of deciding the satisfiability of a system of linear equations modulo 2, with three variables per equation.

We consider the following query model: the algorithm can query if an equation contains a variable. If the answer is \textbf{YES}, then the algorithm is also given all the variables and the right-hand side of the equation.

\begin{theorem} \label{thm:lbeq}
    For any $\eps>0$, any algorithm that distinguishes between a perfectly satisfiable $\Equ$ instance and an instance that satisfies at most $(1/2+\eps)$-fraction of equations needs $\Omega(n^2)$ queries with probability at least $2/3$.
\end{theorem}

We start by defining the hard distribution. The distribution is similar to the one in \cite{bogdanov2002lower}, but the query model and therefore the proof are different.
Every hard instance has $n$ variables $x_1,x_2,\dots x_n$ and $m=kn$ equations $e_1,e_2,\dots,e_m$ for some positive integer $k$. We construct the following two distributions of $\Equ$. 
\begin{itemize}
    \item The distribution $\dis_{NO}$ is the distribution of NO-instance, and is generated as follows: We first generate a random permuation $\sigma : [1,3m] \rightarrow [1,3m]$. For each $1\le i \le m$, we assign equation $e_i$ the variables $y^i_1=x_{\cei{\frac{\sigma(3i-2)}{3k}}}$, $y^i_2=x_{\cei{\frac{\sigma(3i-1)}{3k}}}$ and $y^i_3=x_{\cei{\frac{\sigma(3i)}{3k}}}$. The equation $e_i$ is $y^i_1+y^i_2+y^i_3 = \pm z_i$ where $z_i$ is choosen to be $+1$ or $-1$ uniformly randomly.
    \item The  distribution $\dis_{YES}$ is the distribution of YES-instance, and is generated as follows: We first assign the variables to each equation with the same process as $\dis_{NO}$. Then we randomly choose an assignment of vaiables, say $A^{\star}$. Finally, for each equation $e_i$, we set $y^i_1+y^i_2+y^i_3 = z_i$ where $z_i$ equals the sum of $y^i_1+y^i_2+y^i_3$ according to assignment $A^{\star}$.
\end{itemize}

Our final distribution generates an instance from the NO-distribution with probability $1/2$ and an instance from the YES-distribution with probability $1/2$.

If the instance is generated by $\dis_{YES}$, then it is satisfied by the assignment $A^{\star}$. The following lemma proves that if the instance is generated by $\dis_{NO}$, then with high probability, the at most $(1/2+\eps)$-fraction of the equations can be satisfied.

\begin{lemma} \label{lem:eqno}
    For any $\eps>0$, there exists a positive integer $k$, such that if an instance of $\Equ$ is randomly chosen from $\dis_{NO}$ with $n$ variables and $m = kn$ equations, then with probability $9/10$, at most $(1/2+\eps)$-fraction of the equations can be satisfied.
\end{lemma}

\begin{proof}
    Let $k = 8/\eps^2$ and so $m=\frac{8n}{\eps^2}$. Fix an assignment $A$. For each equation $e_i$, the probability that $A$ satisfies $e_i$ is $1/2$. Since in distribution $\dis_{NO}$, the right hand side of the equations are sampled independently, the event that $A$ satisfies any equation is independent of the event of $A$ satisfying any subset of the other equations. By the Chernoff bound, the probability that $A$ satisfies at least $(1/2+\eps)$-fraction of equations is at most $e^{-\frac{\eps^2 (m/2)}{4}} \le e^{-n}$. Taking the union bound over all possible assignments $A$, the probability that there exists an assignment that satisfies at least $(1/2+\eps)$-fraction of equations is at most $2^n \cdot e^{-n} < 1/10$.
\end{proof}

Now we prove that it is hard to distinguish between the YES and NO instances of this distributions. Define a bipartite graph $G_{\sigma}$  associated with the random permutation $\sigma$ as follows: there are $3m$ vertices on each side of $G_{\sigma}$, there is an edge between the $i^{th}$ vertex on the left and the $j^{th}$ vertex on the right if and only if $\sigma_i = j$. Since $\sigma$ is chosen uniformly at random, $G_{\sigma}$ is a randomly chosen perfect matching. Associate variable $x_i$ with the $(3k(i-1)+1)^{th}$ to the $(3ki)^{th}$ vertices on the left and associate equation $e_j$ with the $(3j-2)^{th}$ to the $(3j)^{th}$ vertex on the right. A variable occurs in an equation if and only if there is an edge between the vertices associate with the variable and the equation.

Fix an algorithm $\Alg$, let $\EE^{\Alg}_{YES}$ and $\EE^{\Alg}_{NO}$ be the set of equations given to $\Alg$ after all the queries to an instance generated by $\DYS$ and $\DNS$. Denote the knowledge graph $G^{\Alg}$ as the subgraph of $G_{\sigma}$ induced by the equations given to $\Alg$ and the variables that occur in these equations. The following lemma shows that if an algorithm only discover a small fraction of equations, then the set of equations discovered by the algorithm has the same distribution in the YES and NO cases with some high constant probability.

\begin{lemma} \label{lem:eqind}
    For any $k>0$, there exists a constant $\delta_0$ such that: if $G^{\Alg}$ contains at most $3 \delta_0 n$ edges, then the distributions of $\EE^{\Alg}_{YES}$ and $\EE^{\Alg}_{NO}$ are identical with probability at least $9/10$.
\end{lemma}

The proof of Lemma~\ref{lem:eqind} is similar to the proof of Theorem~8 in \cite{bogdanov2002lower}. We prove that the left hand side of the equations in $\EE^{\Alg}_{YES}$ and $\EE^{\Alg}_{NO}$ are independent, and thus the distribution of the right hand side are identical.

\begin{proof}
    We first prove that there is a constant $\delta_0$ such that with probability at least $9/10$, any set of equations of size $\delta n \le \delta_0 n$ contains more than $\frac{3}{2}\delta n$ variables. Fix a set of variables $V$ of size $\frac{3}{2}\delta n$. For any equation $e$, the probability that it contains only the variables in $V$ is $\frac{4.5 k \delta n}{3 k n} \cdot \frac{4.5 k\delta n-1}{3kn-1} \cdot \frac{4.5 k \delta n -2}{3kn-2} \le 4 \delta^3$. For any equation $e$ and any set of equations $\EE$ that does not contain $e$, the events that $e$ only contains variables in $V$ and the equations in $\EE$ only contain variable in $V$ are negatively correlated. So for any set of equations of size $\delta n$, the probability that these equations only contain variables in $V$ is at most $(4\delta^3)^{\delta n} = 4^{\delta n} \delta^{3\delta n}$. Taking the union bound over all possible set of equations of size $\delta n$, the probability that one of them only contains variables in $V$ is at most $4^{\delta n} \delta^{3\delta n} \cdot \binom{kn}{\delta n} \le 4^{\delta n} \delta^{3\delta n} \cdot (e k/\delta)^{\delta n} = (4ek)^{\delta n} \delta^{2\delta n}$. We now take the union bound over all sets of variables of size $\frac{3}{2}\delta n$; the probability that there exists a set of equations of size $\delta n$ which only contains $\frac{3}{2}\delta n$ variables is at most $(4ek)^{\delta n} \delta^{2\delta n} \cdot \binom{n}{1.5 \delta n} \le (4ek)^{\delta n} \delta^{2\delta n} \cdot (\frac{2e}{3\delta})^{1.5\delta n} \le (3e^{2.5}k)^{\delta n} \delta^{0.5\delta n} = (40k \sqrt{\delta})^{\delta n} \le (40k \sqrt{\delta_0})^{\delta n}$. Let $\delta_0 < \frac{1}{11 (40k)^2}$, and taking union bound over all possible sizes $i$ ranging from $1$ to $\delta_0 n$, the probability that any set of equations of size $i \le \delta_0 n$ contains more than $\frac{3}{2} i$ variables is at least $1 - \sum_{i=1}^{\delta_0 n} (\frac{1}{11})^i \ge 9/10$. 

    So with probability at least $9/10$, any set of equations with size $i \le \delta_0 n$ contains more than $\frac{3}{2}i$ variables, which means there is at least one variable that occurs at most once in these equations by the pigeonhole principle. We prove that under this event, the distribution of $\EE^{\Alg}_{YES}$ and $\EE^{\Alg}_{NO}$ are identical if $G^{\Alg}$ contains at most $3 \delta_0 n$ edges. 

    Notice that the left hand side of of $\EE^{\Alg}_{YES}$ and $\EE^{\Alg}_{NO}$ are always identical, we only need to prove that the distributions of the right hand side are identical when $G^{\Alg}$ has at most $3 \delta_0 n$ edges. In this case there are at most $\delta_0 n$ equations in $\EE^{\Alg}_{YES}$ since each equation is associated with 3 vertices. Let the right hand sides of $\EE^{\Alg}_{YES}$ and $\EE^{\Alg}_{NO}$ be vectors $b_{YES}$ and $b_{NO}$ respectively. We prove the distributions of $b_{YES}$ and $b_{NO}$ are identical by induction on the size of $b_{YES}$ (which is also the number of equations in $\EE^{\Alg}_{YES}$). 

    The base case is when there is no equation in $\EE^{\Alg}_{YES}$ at all (which means the algorithm does not discover any equation). In this case, both $b_{YES}$ and $b_{NO}$ are empty vectors.

    In the induction step, $\card{b_{YES}} = \card{b_{NO}} > 0$. Since the number of equations is at most $\delta_0 n$, there exists a variable $v$ that only occurs once. Without loss of generality, suppose it occurs in the last equation. Let $b'_{YES}$ and $b'_{NO}$ be the vector obtained by deleting the last entry of $b_{YES}$ and $b_{NO}$ respectively. By induction hypothesis, the distributions of $b'_{YES}$ and $b'_{NO}$ are identical. Moreover, $v$ only occurs in the last equation and only occurs once, the distribution of the last entry of $b_{YES}$ is uniform, independent of the other entries, so is the last entry of $d_{NO}$. So the distributions of $b_{YES}$ and $b_{NO}$ are identical.
\end{proof}

Next we prove that in order to discover a constant fraction of equations, we need $\Omega(n^2)$ queries.

\begin{lemma} \label{lem:eqquery} 
    For any $\delta_0>0$, there exists a $\delta_1>0$ such that: for any algorithm that makes at most $\delta_1 n^2$ queries, $G^{\Alg}$ contains at most $3 \delta_0 n$ edges with probability $9/10$.
\end{lemma}

The proof of Lemma~\ref{lem:eqquery} is similar to the proof of Theorem~5.2 in \cite{assadi2019sublinear} and we will prove it later.

\begin{proof}[Proof of Theorem~\ref{thm:lbeq}]
    For any $\eps>0$, let $k,\delta_0,\delta_1$ be the constant defined in Lemma~\ref{lem:eqno}, Lemma~\ref{lem:eqind}, Lemma~\ref{lem:eqquery} respectively. Consider two instance $I_{YES}$ and $I_{NO}$ generated as follows: we generate the instance $I_{YES}$ by distribution $\dis_{YES}$, then let the left hand side of $I_{NO}$ be the same as the left hand side of $I_{YES}$, generate the right hand side of $I_{NO}$ uniformly independently for each equation. Since the process of generating the left hand side is the same for $\dis_{YES}$ and $\dis_{NO}$, the distribution of $I_{NO}$ is indeed $\dis_{NO}$. By Lemma~\ref{lem:eqno}, with probability $9/10$, the $I_{NO}$ satisfies at most $(1/2+\eps)$-fraction of equations. By Lemma~\ref{lem:eqquery}, if an algorithm makes at most $\delta_1 n^2$ queries, then it discovers at most $\delta_0 n$ equations with probability $9/10$. Base on this event, by Lemma~\ref{lem:eqind}, the equations discovered by the algorithm has the same probability of being generated by $\dis_{YES}$ and  by $\dis_{NO}$. By the union bound, with probability at most $7/10$, $I_{NO}$ is an instance that satisfies at most $(1/2+\eps)$-fraction of the equations and the algorithm cannot distinguish between $I_{YES}$ and $I_{NO}$.
\end{proof}

We use the following standard reduction from Equation to 3SAT in \cite{haastad2001some}.
Given a set of equations $\EE$, we construct a 3CNF formula $\Phi = F(\EE)$ as follows: For any equation $X_i+X_j+X_k=1$ in $\EE$, we add four clauses $(X_i \vee X_j \vee X_k)$, $(X_i \vee \bar{X_j} \vee \bar{X_k})$, $(\bar{X_i} \vee X_j \vee \bar{X_k})$ and $(\bar{X_i} \vee \bar{X_j} \vee X_k)$ into $\Phi$; for any equation $X_i+X_j+X_k=0$ in $\EE$, we add four clauses $(\bar{X_i} \vee X_j \vee X_k)$, $(X_i \vee \bar{X_j} \vee X_k)$, $(X_i \vee X_j \vee \bar{X_k})$ and $(\bar{X_i} \vee \bar{X_j} \vee \bar{X_k})$ into $\Phi$. It is clear that if an assigment satisfies an equation in $\EE$, then it also satisfies all of the four corresponding clauses in $\Phi$. Otherwise it satisfies three of the four corresponding clauses. So we have the following lemma.

\begin{lemma}  \label{lem:red-eq-3sat}
    For any $\gamma \in (0,1]$, given a set of equations $\EE$ and its corresponding 3CNF formula $\Phi = F(\EE)$, for any assignment $A$, $A$ satisfies $\gamma$-fraction of equations in $\EE$ if and only if $A$ satisfies $(3/4+\gamma/4)$-fraction of clauses in $\Phi$.
\end{lemma}

\begin{proof} [Proof of Theorem~\ref{thm:3SAT}]
    For any 3CNF formula generated from a $\Equ$ instance $\EE$ we consider a stronger type of query model for 3SAT. For any query between a variable and a clause, if the variable occurs in the clause, then the algorithm is not only given the entire clause, but also the other 3 clauses corresponding to the same equation in $\EE$. The new query is equivalent to the query in $\Equ$.
\end{proof}

\subsection{Proof of Lemma~\ref{lem:eqquery}} \label{sec:app-lb}

The proof of Lemma~\ref{lem:eqquery} is similar to the proof of Theorem~5.2 in \cite{assadi2019sublinear}. However, the argument from \cite{assadi2019sublinear} cannot be used in a black-box manner. So, here we present  a complete proof. The following lemma from \cite{assadi2019sublinear} is useful.

\begin{lemma} [Lemma~5.4 of \cite{assadi2019sublinear}]
\label{lem:random_matching}
Let $G'(L' \cup R',E')$ be an arbitrary bipartite graph such that $|L'| = |R'| = N$, and each vertex in $G'$ has degree at least $2N/3$. Then for any edge $e = (u,v) \in E'$, the probability that $e$ is contained in a  perfect matching chosen uniformly at random in $G'$ is at most $3/N$. 
\end{lemma}

Denote the vertex sets in bipartite graph $G_{\sigma}$ as $L$ and $R$. We have $\card{L}=\card{R} = 3kn$. Suppose whenever a query finds a variable inside an equation, the algorithm is not only given the equation, but also edges incident on the vertices associated with the equation in $G_{\sigma}$. Then, $G^{\Alg}$ contains exactly  those edges that are given to algorithm in response to the queries. 

The query process can be viewed as the task of finding $\delta_0 n$ edges in $G_{\sigma}$ by the following queries between a variable $x_i$ and an equation $e_j$: query if there is at least one edge between $U$ and $V$ where $U \subset L$ is the set of vertices associated with $x_i$ and $V \subset R$ is the set of vertices associated with $e_j$. If so, the algorithm is given all edges  incident on the vertices in $V$. To prove the lemma, we only need to prove that finding $3 \delta_0 n$ edges in $G_{\sigma}$ in this model needs $\Omega(n^2)$ queries. 

For simplicity, we consider the following query model instead: a query asks if there is an edge between a pair of vertices $u$ and $v$. If so, the algorithm is given all three edges incident on the vertices associated with the same equation as $v$. Any original query can be simulated by $3k \cdot 3 = 9k$ new queries. So it is sufficient to prove that we need $\Omega(n^2)$ queries in the new model.

We say that an edge $(u,v)$ in $G_{\sigma}$ has been {\em discovered} if the edge is given to the algorithm.
After $t$ queries have been made by the algorithm, let $L_U(t) \subseteq L$ and $R_U(t) \subseteq R$ denote the set of undiscovered vertices in $L$ and $R$ respectively.  Let $E(t) \subseteq L_U(t) \times R_U(t)$ denote the set of edge slots that have not yet been queried/discovered. Note that by our process for generating $G_{\sigma}$, the undiscovered edges correspond to a random perfect matching between $L_U(t)$ and $R_U(t)$ that is entirely supported on $E(t)$. 

We will analyze the performance of any algorithm by partitioning the queries into phases. The first query by the algorithm starts the first phase, and a phase ends as soon as three edges in $G_{\sigma}$ have been discovered. 
Let $Z_i$ be a random variable that denotes the number of queries performed in phase $i$ of the algorithm. Thus we wish to analyze $\expect{\sum_{i=1}^{\delta_0 n} Z_i}$. 

For any vertex $w \in (L_U(t) \cup R_U(t))$, we say that the {\em uncertainty} of $w$ is $d$ if there are at least $d$ edge slots in $E(t)$ that are incident on $w$.

At time $t$, we say a vertex $w \in (L_U(t) \cup R_U(t))$ is {\em bad} if the uncertainty of $w$ is less than $2.5kn$. 
Note that if at some time $t$ none of the vertices in $(L_U(t) \cup R_U(t))$ are bad then
in the next $nk/2$ time steps, 
   the degree of any vertex in $L_U(t) \cup R_U(t)$ in $E(t)$ remains above $2kn$ if there is no successful query. Thus by Lemma~\ref{lem:random_matching}, the probability that any query made during the first $nk/2$ queries in the phase succeeds (in discovery of a new edge in $G_{\sigma}$) is at most $3/(3kn) = 1/(kn)$. 

   We say a phase is {\em good} if at the \underline{start} of the phase, there are no bad vertices, and the phase is {\em bad} otherwise.

   \begin{proposition}
   The expected length of a good phase is at least $nk/4$.
   \end{proposition}
   \begin{proof}
   If at the start of the phase $i$, no vertex is bad, then for the next $nk/2$ time steps, the probability of success for any query is at most $1/(kn)$. Thus the expected number of successes (discovery of a new edge in $G_{\sigma}$) in the first $nk/2$ time steps in a phase is at most $1/2$. By Markov's inequality, it then follows that with probability at least $1/2$, there are no successes among the first $nk/2$ queries in a phase. Thus the expected length of the phase is  $\geq nk/4$. 
   \end{proof}

   Note that if all phases were good, then it immediately follows that the expected number of queries to discover $3 \delta_0 n$ edges is $\Omega(n^2)$. To complete the proof, it remains to show that most phases are good. For ease of analysis, we will give the algorithm additional information for free and show that it still needs $\Omega(n^2)$ queries in expectation even to discover the first $3 \delta_0 n$ edges in $G_{\sigma}$.

   Whenever the algorithm starts a bad phase, we immediately reveal to the algorithm an undiscovered edge $(u,v)$ in $G_{\sigma}$ (as well as other two edges incident on the vertices associated with the same equation as $v$) that is incident on an arbitrarily chosen bad vertex. Thus each bad phase is guaranteed to consume a bad vertex (i.e., make the bad vertex discovered and hence remove it from further consideration). On the other hand, to create a bad vertex $w$, one of the following two events needs to occur:
   the number of discovered edges in $G_{\sigma}$ plus the number of queries is at least $3kn - 2.5kn = kn/2$.

   Since we are restricting ourselves to analyzing the discovery of first $\delta_0 n$ edges in $G_{\sigma}$, any vertex $w$ that becomes bad requires at least $kn/2$ queries incident on it. Thus to create $K$ bad vertices in the first $\delta_0 n$ phases, we 
   need to perform at least $(K \cdot (kn/2))/2$ queries; here the division by $2$ accounts for the fact that each query reduces uncertainty for two vertices. It now follows that if the algorithm encounters at least $\delta_0 n/2$ bad phases among the first $\delta_0$ 
   phases, then $K \ge \delta_0 n/2$ and hence it must have already performed $\delta_0 k n^2/8$ queries. Otherwise, at least $\delta_0 n/2$ phases among the first $\delta_0$ phases are good, implying that the expected number of queries is at least $(\delta_0 n/2)  \cdot (nk/4) = \Omega(n^2)$. This completes 
   the proof of Lemma~\ref{lem:eqquery} with Markov's inequality.

\section{A Reduction from Matching Size to TSP Cost Estimation}
\label{sec:tsp_vs_matchingsize}

In this section, we give a reduction from the problem of estimating the maximum matching size in a bipartite graph to the problem of estimating the optimal $(1,2)$-TSP cost. An essentially identical reduction works for graphic TSP cost using the idea described in Section~\ref{sec:lb_graphic_TSP}.

We will denote the size of the largest matching in a graph $G$ by $\alpha(G)$.
Given a bipartite graph $G(V,E)$ with $n$ vertices on each side, we construct an instance $G'(V', E')$ of the $(1,2)$-TSP problem on $4n$ vertices such that the optimal TSP cost on $G'$ is $5n - \alpha(G)$.
Thus for any $\eps \in [0, 1/5)$, any algorithm that can estimate $(1,2)$-TSP cost to within a $(1+\eps)$-factor, also gives us an estimate of the matching size in $G$ to within an additive error of $5 \eps n$.

We will now describe our construction of the graph $G'$. For clarity of exposition, we will describe $G'$ as the graph that contains edges of cost $1$ -- all other edges have cost $2$.
Suppose the vertex set $V$ of $G$ consists of the bipartition
$V_1=\{v^1_1,v^1_2,\dots,v^1_n\}$ and $V_2=\{v^2_1,v^2_2,\dots,v^2_n\}$. We construct the graph $G'$ as follows: we start with the graph $G$, then add three sets of vertices $V_0$, $V_3$ and $V_4$, such that $V_0 = \{v^0_1,v^0_2,\dots,v^0_{n/2}\}$ with $n/2$ vertices, $V_3=\{v^3_1,v^3_2,\dots,v^3_n\}$ with $n$ vertices, and $V_4=\{v^4_1,v^4_2,\dots,v^4_{n/2}\}$ with $n/2$ vertices. The graph $G'$ will only have edges between $V_j$ and $V_{j+1}$ ($j = \{0,1,2,3\})$. We will denote the set of edges between $V_j$ and $V_{j+1}$ as $E_{j,j+1}$. For any vertex $v^0_i \in V_0$, it connects to $v^1_{2i-1}$ and $v^1_{2i}$ in $V_1$. $E_{1,2}$ has the same edges as the edges in $G$. Each vertex $v^2_i \in V_2$ is connected to vertex $v^3_i$ in $V_3$, that is, vertices in $V_2$ and $V_3$ induce a perfect matching (identity matching). Finally, each vertex in $V_3$ is connected to all the vertices in $V_4$. See Figure~\ref{fig:reduction}(a) for an illustration.

\begin{figure}[h!]
    \centering
        \subcaptionbox{ The illustration of $G'$.
    }[0.47 \textwidth]{ 
\begin{tikzpicture}[ auto ,node distance =1cm and 2cm , on grid , semithick , state/.style ={ circle ,top color =white , bottom color = white , draw, black , text=black}, every node/.style={inner sep=0,outer sep=0}]

\node(a0){};
\node[state, circle, black, line width=0.15mm, minimum height=7pt, minimum width=7pt] (a1) [below= 0.375cm of a0]{};
\node[state, circle, black, line width=0.15mm, minimum height=7pt, minimum width=7pt] (a2) [below= 1.5cm of a1]{};
\node[state, circle, black, line width=0.15mm, minimum height=7pt, minimum width=7pt] (a3) [below=1.5cm of a2]{};
\node (v0) [below=0.75cm of a3]{$V_0$};

%%\node[state, circle, black, line width=0.15mm, minimum height=7pt, minimum width=7pt] (a7) [below=0.75cm of a6]{};
%%\node[state, circle, black, line width=0.15mm, minimum height=7pt, minimum width=7pt] (a8) [below=0.75cm of a7]{};

\node[rectangle, rounded corners = 3mm, inner sep=5pt, draw,  black, fit=(a1) (a3), line width=0.15mm] {};

\node[state, circle, black, draw, line width=0.15mm, minimum height=7pt, minimum width=7pt] (b1) [right=1.5cm of a0]{};
\node[state, circle, black, line width=0.15mm, minimum height=7pt, minimum width=7pt] (b2) [below= 0.75cm of b1]{};
\node[state, circle, black, line width=0.15mm, minimum height=7pt, minimum width=7pt] (b3) [below=0.75cm of b2]{};
\node[state, circle, black, line width=0.15mm, minimum height=7pt, minimum width=7pt] (b4) [below=0.75cm of b3]{};
\node[state, circle, black, line width=0.15mm, minimum height=7pt, minimum width=7pt] (b5) [below=0.75cm of b4]{};
\node[state, circle, black, line width=0.15mm, minimum height=7pt, minimum width=7pt] (b6) [below=0.75cm of b5]{};
\node (v1) [below=0.75cm of b6]{$V_1$};

%%\node[state, circle, black, line width=0.15mm, minimum height=7pt, minimum width=7pt] (b7) [below=0.75cm of b6]{};
%%\node[state, circle, black, line width=0.15mm, minimum height=7pt, minimum width=7pt] (b8) [below=0.75cm of b7]{};

\node[rectangle,rounded corners = 3mm, inner sep=5pt, draw,  black, fit=(b1) (b6), line width=0.15mm] {};

\node[state, circle, black, draw, line width=0.15mm, minimum height=7pt, minimum width=7pt] (c1) [right=1.5cm of b1]{};
\node[state, circle, black, line width=0.15mm, minimum height=7pt, minimum width=7pt] (c2) [below= 0.75cm of c1]{};
\node[state, circle, black, line width=0.15mm, minimum height=7pt, minimum width=7pt] (c3) [below=0.75cm of c2]{};
\node[state, circle, black, line width=0.15mm, minimum height=7pt, minimum width=7pt] (c4) [below=0.75cm of c3]{};
\node[state, circle, black, line width=0.15mm, minimum height=7pt, minimum width=7pt] (c5) [below=0.75cm of c4]{};
\node[state, circle, black, line width=0.15mm, minimum height=7pt, minimum width=7pt] (c6) [below=0.75cm of c5]{};
\node (v2) [below=0.75cm of c6]{$V_2$};
%%\node[state, circle, black, line width=0.15mm, minimum height=7pt, minimum width=7pt] (c7) [below=0.75cm of c6]{};
%%\node[state, circle, black, line width=0.15mm, minimum height=7pt, minimum width=7pt] (c8) [below=0.75cm of c7]{};

\node[rectangle,rounded corners = 3mm, inner sep=5pt, draw,  black, fit=(c1) (c6), line width=0.15mm] {};

\node[state, circle, black, draw, line width=0.15mm, minimum height=7pt, minimum width=7pt] (d1) [right=1.5cm of c1]{};
\node[state, circle, black, line width=0.15mm, minimum height=7pt, minimum width=7pt] (d2) [below= 0.75cm of d1]{};
\node[state, circle, black, line width=0.15mm, minimum height=7pt, minimum width=7pt] (d3) [below=0.75cm of d2]{};
\node[state, circle, black, line width=0.15mm, minimum height=7pt, minimum width=7pt] (d4) [below=0.75cm of d3]{};
\node[state, circle, black, line width=0.15mm, minimum height=7pt, minimum width=7pt] (d5) [below=0.75cm of d4]{};
\node[state, circle, black, line width=0.15mm, minimum height=7pt, minimum width=7pt] (d6) [below=0.75cm of d5]{};
\node (v3) [below=0.75cm of d6]{$V_3$};
%%\node[state, circle, black, line width=0.15mm, minimum height=7pt, minimum width=7pt] (d7) [below=0.75cm of d6]{};
%%\node[state, circle, black, line width=0.15mm, minimum height=7pt, minimum width=7pt] (d8) [below=0.75cm of d7]{};

\node[rectangle,rounded corners = 3mm, inner sep=5pt, draw,  black, fit=(d1) (d6), line width=0.15mm] {};

\node(e0) [right=1.5cm of d1]{};
\node[state, circle, black, line width=0.15mm, minimum height=7pt, minimum width=7pt] (e1) [below= 0.375cm of e0]{};
\node[state, circle, black, line width=0.15mm, minimum height=7pt, minimum width=7pt] (e2) [below=1.5cm of e1]{};
\node[state, circle, black, line width=0.15mm, minimum height=7pt, minimum width=7pt] (e3) [below=1.5cm of e2]{};
\node (v4) [below=0.75cm of e3]{$V_4$};
%%\node[state, circle, black, line width=0.15mm, minimum height=7pt, minimum width=7pt] (d7) [below=0.75cm of d6]{};
%%\node[state, circle, black, line width=0.15mm, minimum height=7pt, minimum width=7pt] (d8) [below=0.75cm of d7]{};

\node[rectangle,rounded corners = 3mm, inner sep=5pt, draw,  black, fit=(e1) (e3), line width=0.15mm] {};

\draw[-, black] (a1) to (b1) node[near start]{}; 
\draw[-, black] (a1) to (b2) node[near start]{}; 
\draw[-, black] (a2) to (b3) node[near start]{}; 
\draw[-, black] (a2) to (b4) node[near start]{}; 
\draw[-, black] (a3) to (b5) node[near start]{}; 
\draw[-, black] (a3) to (b6) node[near start]{}; 

\draw[-, black] (b1) to (c3) node[near start]{}; 
\draw[-, black] (b1) to (c5) node[near start]{}; 
\draw[-, black] (b2) to (c5) node[near start]{}; 
\draw[-, black] (b3) to (c1) node[near start]{}; 
\draw[-, black] (b3) to (c6) node[near start]{}; 
\draw[-, black] (b4) to (c4) node[near start]{}; 
\draw[-, black] (b4) to (c2) node[near start]{}; 
\draw[-, black] (b5) to (c5) node[near start]{}; 
\draw[-, black] (b6) to (c3) node[near start]{}; 
\draw[-, black] (b6) to (c1) node[near start]{};

\draw[-, black] (c1) to (d1) node[near start]{}; 
\draw[-, black] (c2) to (d2) node[near start]{}; 
\draw[-, black] (c3) to (d3) node[near start]{}; 
\draw[-, black] (c4) to (d4) node[near start]{}; 
\draw[-, black] (c5) to (d5) node[near start]{}; 
\draw[-, black] (c6) to (d6) node[near start]{}; 

\draw[-, black] (d1) to (e1) node[near start]{}; 
\draw[-, black] (d2) to (e1) node[near start]{}; 
\draw[-, black] (d3) to (e1) node[near start]{}; 
\draw[-, black] (d4) to (e1) node[near start]{}; 
\draw[-, black] (d5) to (e1) node[near start]{}; 
\draw[-, black] (d6) to (e1) node[near start]{}; 
\draw[-, black] (d1) to (e2) node[near start]{}; 
\draw[-, black] (d2) to (e2) node[near start]{}; 
\draw[-, black] (d3) to (e2) node[near start]{}; 
\draw[-, black] (d4) to (e2) node[near start]{}; 
\draw[-, black] (d5) to (e2) node[near start]{}; 
\draw[-, black] (d6) to (e2) node[near start]{}; 
\draw[-, black] (d1) to (e3) node[near start]{}; 
\draw[-, black] (d2) to (e3) node[near start]{}; 
\draw[-, black] (d3) to (e3) node[near start]{}; 
\draw[-, black] (d4) to (e3) node[near start]{}; 
\draw[-, black] (d5) to (e3) node[near start]{}; 
\draw[-, black] (d6) to (e3) node[near start]{};

\end{tikzpicture}
}\label{fig:red}
\hspace{2mm}  
        \subcaptionbox{ The illustration of tour $T$, where $V_2$ and $V_3$ are arranged with order $(v^2_{f(1)},\dots,v^2_{f(6)})$ and $(v^3_{f(1)},\dots,v^3_{f(6)})$.
    }[0.49 \textwidth]{  
\begin{tikzpicture}[ auto ,node distance =1cm and 2cm , on grid , semithick , state/.style ={ circle ,top color =white , bottom color = white , draw, black , text=black}, every node/.style={inner sep=0,outer sep=0}]

\node(a0){};
\node[state, circle, black, line width=0.15mm, minimum height=7pt, minimum width=7pt] (a1) [below= 0.375cm of a0]{};
\node[state, circle, black, line width=0.15mm, minimum height=7pt, minimum width=7pt] (a2) [below= 1.5cm of a1]{};
\node[state, circle, black, line width=0.15mm, minimum height=7pt, minimum width=7pt] (a3) [below=1.5cm of a2]{};
\node (v0) [below=0.75cm of a3]{$V_0$};

%%\node[state, circle, black, line width=0.15mm, minimum height=7pt, minimum width=7pt] (a7) [below=0.75cm of a6]{};
%%\node[state, circle, black, line width=0.15mm, minimum height=7pt, minimum width=7pt] (a8) [below=0.75cm of a7]{};

\node[rectangle, rounded corners = 3mm, inner sep=5pt, draw,  black, fit=(a1) (a3), line width=0.15mm] {};

\node[state, circle, black, draw, line width=0.15mm, minimum height=7pt, minimum width=7pt] (b1) [right=1.5cm of a0]{};
\node[state, circle, black, line width=0.15mm, minimum height=7pt, minimum width=7pt] (b2) [below= 0.75cm of b1]{};
\node[state, circle, black, line width=0.15mm, minimum height=7pt, minimum width=7pt] (b3) [below=0.75cm of b2]{};
\node[state, circle, black, line width=0.15mm, minimum height=7pt, minimum width=7pt] (b4) [below=0.75cm of b3]{};
\node[state, circle, black, line width=0.15mm, minimum height=7pt, minimum width=7pt] (b5) [below=0.75cm of b4]{};
\node[state, circle, black, line width=0.15mm, minimum height=7pt, minimum width=7pt] (b6) [below=0.75cm of b5]{};
\node (v1) [below=0.75cm of b6]{$V_1$};

%%\node[state, circle, black, line width=0.15mm, minimum height=7pt, minimum width=7pt] (b7) [below=0.75cm of b6]{};
%%\node[state, circle, black, line width=0.15mm, minimum height=7pt, minimum width=7pt] (b8) [below=0.75cm of b7]{};

\node[rectangle,rounded corners = 3mm, inner sep=5pt, draw,  black, fit=(b1) (b6), line width=0.15mm] {};

\node[state, circle, black, draw, line width=0.15mm, minimum height=7pt, minimum width=7pt] (c1) [right=1.5cm of b1]{};
\node[state, circle, black, line width=0.15mm, minimum height=7pt, minimum width=7pt] (c2) [below= 0.75cm of c1]{};
\node[state, circle, black, line width=0.15mm, minimum height=7pt, minimum width=7pt] (c3) [below=0.75cm of c2]{};
\node[state, circle, black, line width=0.15mm, minimum height=7pt, minimum width=7pt] (c4) [below=0.75cm of c3]{};
\node[state, circle, black, line width=0.15mm, minimum height=7pt, minimum width=7pt] (c5) [below=0.75cm of c4]{};
\node[state, circle, black, line width=0.15mm, minimum height=7pt, minimum width=7pt] (c6) [below=0.75cm of c5]{};
\node (v2) [below=0.75cm of c6]{$V_2$};
%%\node[state, circle, black, line width=0.15mm, minimum height=7pt, minimum width=7pt] (c7) [below=0.75cm of c6]{};
%%\node[state, circle, black, line width=0.15mm, minimum height=7pt, minimum width=7pt] (c8) [below=0.75cm of c7]{};

\node[rectangle,rounded corners = 3mm, inner sep=5pt, draw,  black, fit=(c1) (c6), line width=0.15mm] {};

\node[state, circle, black, draw, line width=0.15mm, minimum height=7pt, minimum width=7pt] (d1) [right=1.5cm of c1]{};
\node[state, circle, black, line width=0.15mm, minimum height=7pt, minimum width=7pt] (d2) [below= 0.75cm of d1]{};
\node[state, circle, black, line width=0.15mm, minimum height=7pt, minimum width=7pt] (d3) [below=0.75cm of d2]{};
\node[state, circle, black, line width=0.15mm, minimum height=7pt, minimum width=7pt] (d4) [below=0.75cm of d3]{};
\node[state, circle, black, line width=0.15mm, minimum height=7pt, minimum width=7pt] (d5) [below=0.75cm of d4]{};
\node[state, circle, black, line width=0.15mm, minimum height=7pt, minimum width=7pt] (d6) [below=0.75cm of d5]{};
\node (v3) [below=0.75cm of d6]{$V_3$};
%%\node[state, circle, black, line width=0.15mm, minimum height=7pt, minimum width=7pt] (d7) [below=0.75cm of d6]{};
%%\node[state, circle, black, line width=0.15mm, minimum height=7pt, minimum width=7pt] (d8) [below=0.75cm of d7]{};

\node[rectangle,rounded corners = 3mm, inner sep=5pt, draw,  black, fit=(d1) (d6), line width=0.15mm] {};

\node(e0) [right=1.5cm of d1]{};
\node[state, circle, black, line width=0.15mm, minimum height=7pt, minimum width=7pt] (e1) [below= 0.375cm of e0]{};
\node[state, circle, black, line width=0.15mm, minimum height=7pt, minimum width=7pt] (e2) [below=1.5cm of e1]{};
\node[state, circle, black, line width=0.15mm, minimum height=7pt, minimum width=7pt] (e3) [below=1.5cm of e2]{};
\node (v4) [below=0.75cm of e3]{$V_4$};
%%\node[state, circle, black, line width=0.15mm, minimum height=7pt, minimum width=7pt] (d7) [below=0.75cm of d6]{};
%%\node[state, circle, black, line width=0.15mm, minimum height=7pt, minimum width=7pt] (d8) [below=0.75cm of d7]{};

\node (t1) [right=1cm of e2]{};

\node[rectangle,rounded corners = 3mm, inner sep=5pt, draw,  black, fit=(e1) (e3), line width=0.15mm] {};

\draw[-, black] (a1) to (b1) node[near start]{}; 
\draw[-, black] (a1) to (b2) node[near start]{}; 
\draw[-, black] (a2) to (b3) node[near start]{}; 
\draw[-, black] (a2) to (b4) node[near start]{}; 
\draw[-, black] (a3) to (b5) node[near start]{}; 
\draw[-, black] (a3) to (b6) node[near start]{}; 

\draw[-, black] (b1) to (c1) node[near start]{}; 
\draw[-, black] (b2) to (c2) node[near start]{}; 
\draw[-, black] (b3) to (c3) node[near start]{}; 
\draw[-, black] (b4) to (c4) node[near start]{}; 
\draw[-, black] (b5) to (c5) node[near start]{}; 
\draw[-, black] (b6) to (c6) node[near start]{}; 

\draw[-, black] (c1) to (d1) node[near start]{}; 
\draw[-, black] (c2) to (d2) node[near start]{}; 
\draw[-, black] (c3) to (d3) node[near start]{}; 
\draw[-, black] (c4) to (d4) node[near start]{}; 
\draw[-, black] (c5) to (d5) node[near start]{}; 
\draw[-, black] (c6) to (d6) node[near start]{}; 

\draw[-, black] (d2) to (e1) node[near start]{}; 
\draw[-, black] (d3) to (e1) node[near start]{}; 
\draw[-, black] (d4) to (e2) node[near start]{}; 
\draw[-, black] (d5) to (e2) node[near start]{}; 
\draw[-, black] (d6) to (e3) node[near start]{}; 

\draw[-,in=180, black, rounded corners = 20pt] (d1) -| (t1.center);
\draw[-, black, bend right] (e3) to (t1.center) node[near start]{};

\end{tikzpicture}
} \label{fig:tour}
    \caption{An illustration of the reduction for $n=6$.}
    \label{fig:reduction}

\end{figure}
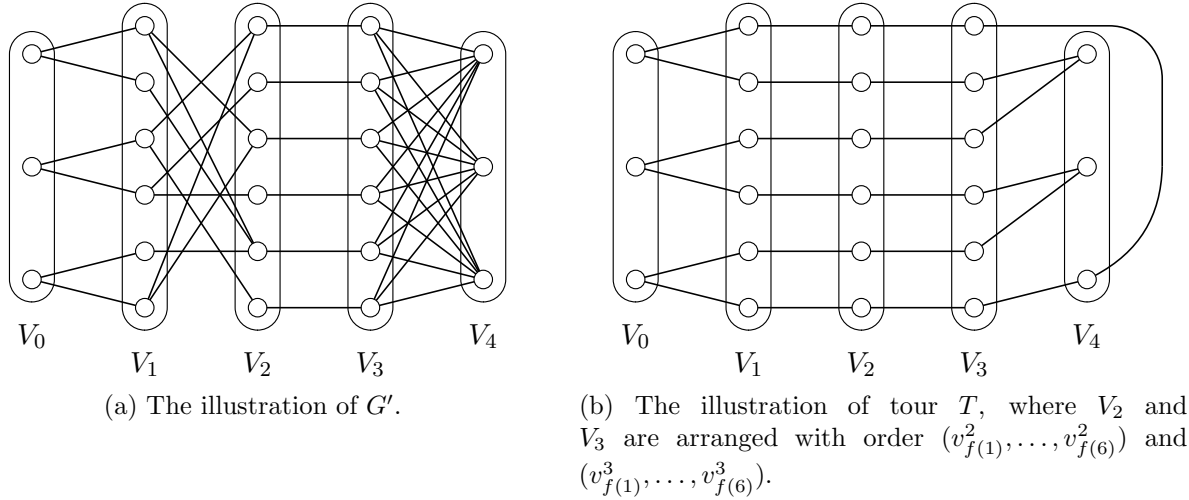

The lemmas below establish a relationship between matching size in $G$ and $(1,2)$-TSP cost in $G'$.

\begin{lemma} \label{lem:red-1}
 Let $M$ be any matching in $G$. Then there is a $(1,2)$-TSP tour $T$ in $G'$ of cost at most $5n-\card{M}$.
\end{lemma}

\begin{proof}
    Let $f: [n] \rightarrow [n]$ be any bijection from $[n]$ to $[n]$ such that whenever a vertex $v^1_i$ is matched to a vertex $v^2_j$ in $M$, then $f(i) = j$. Consider the following $(1,2)$-TSP tour $T$: each vertex $v^0_i \in V_0$ connects to $v^1_{2i-1}$ and $v^1_{2i}$ in $T$; each vertex $v^1_i \in V_1$ connects to $v^0_{\lceil (i+1)/2 \rceil}$ and $v^2_{f(i)}$ in $T$. For any $v^2_{f(i)} \in V_2$, it connects to $v^1_i$ and $v^3_{f(i)}$ in $T$. For any vertex $v^3_{f(i)} \in V_3$, if $i>1$, it connects to $v^2_{f(i)}$ and $v^4_{\ceil{i/2}}$ in $T$; if $i=1$, it connects to $v^2_{f(i)}$ and $v^4_{n/2}$ in $T$. See Figure~\ref{fig:reduction}(b) as an illustration. $T$ is clearly a TSP-tour. 
    
%  {\bf Sanjeev: It should be $i = 1$ above, right? Also, we need to argue that $T$ is connected.}

    All edges in $T$ are also edges in $G'$ except for possibly some edges between $V_1$ and $V_2$. If $v^1_i$ is matched in $M$, then $(v^1_i,v^2_{f(i)})$ is an edge in $G'$, otherwise it is not in $G'$ and thus has weight $2$. So $T$ only has $n-\card{M}$ weight $2$ edges, which means $T$ has cost at most $4n+n-\card{M}=5n-\card{M}$.
\end{proof}

\begin{lemma} \label{lem:red-2}
    For any $(1,2)$-TSP tour $T$ in $G$, $T$ has cost at least $5n-\alpha(G)$.
\end{lemma}

To prove Lemma~\ref{lem:red-2}, we first prove an auxiliary claim.

\begin{claim} \label{cla:red}
    Suppose $G=(V_1,V_2,E)$ is a bipartite graph which has maximum size $\alpha(G)$. For any $2$-degree subgraph $H$ of $G$, if there are at most $X$ vertices in $V_1$ has degree $2$ in $H$, then there are at most $\alpha(G)+X$ vertices in $V_2$ which have degree at least $1$ in $H$. Similarly, if there are at most $X$ vertices in $V_2$ has degree $2$ in $H$, then there are at most $\alpha(G)+X$ vertices in $V_1$ which have degree at least $1$ in $H$.
\end{claim}

\begin{proof}
    If there are at most $X$ vertices in $V_1$ has degree $2$ in $H$. We construct $H'$ by deleting an arbitrary edge on each degree $2$ vertex in $V_1$, then construct $H''$ by deleing an arbitrary edge on each degree $2$ vertex in $V_2$. Since $H''$ does not have degree two vertex, it is a matching of $G$. So the number of degree $1$ vertices in $V_2$ in $H''$ is at most $\alpha(G)$. On the other hand, any vertex in $V_2$ which has degree at least $1$ in $H'$ also has degree $1$ in $H''$. So there are at most $\alpha(G)$ vertices of degree at least $1$ in $V_2$ in $H'$. Furthermore, since there are only $X$ vertices of degree $2$ in $V_1$ in $H$, we delete at most $X$ edges in $H$ when constructing $H'$. So $H'$ has at most $X$ more isolate vertices in $V_2$ than in $H$, which means $H$ has at most $\alpha(G)+X$ vertices with degree at least $1$ in $V_2$.

    The second part of the claim follows via a similar argument as the first part of the claim.
\end{proof}

\begin{proof}[Proof of Lemma~\ref{lem:red-2}]
    Let $a_{01}$ be the number of edges in $T \cap E_{0,1}$, $a_{2,3}$ be the number of edges in $T \cap E_{3,4}$. Let $G^X$ be the intersection graph of $G$ and $T$. Since the vertices in $V_0$ only connect to the vertices in $V_1$ in $G'$, and any vertex in $T$ has degree $2$, there are at least $n-a_{01}$ edges incident on $V_0$ in $T$ are not an edge in $G'$. On the other hand, since any vertex in $V_1$ is incident on at at most $1$ edge in $E_{0,1}$, there are at least $a_{01}$ vertices in $V_1$ is connected to a vertex in $V_0$ in $T$, which means there are at most $n-a_{01}$ vertices in $V_1$ has degree $2$ in $G^X$. By Claim~\ref{cla:red}, there are at most $n-a_{01}+\alpha(G)$ vertices in $V_2$ has edge in $G^X$. For any isolate vertex in $V_2$ in $T$, it has only one edge in $G'$ connecting to $V_3$, so this vertex must incident on an edge in $T$ which is not in $G'$. So there are at least $n-(n-a_{01}+\alpha(G)) = \alpha(G)-a_{01}$ edges incident on $V_2$ in $T$ which is not in $G'$.

    There are $2n$ edges incident on $V_3$ in $T$, but among them, there are only $a_{23}$ edges between $V_2$ and $V_3$ which is also in $G'$, and there are at most $n$ edges between $V_3$ and $V_4$ in $T$ since each vertex has degree only $2$. So there are at least $2n-n-a_{23}=n-a_{23}$ edges incident on $V_3$ which is not in $G'$. On the other hand, since any vertex in $V_2$ is incident on at at most $1$ edge in $E_{2,3}$, there are at least $a_{23}$ vertices in $V_2$ is connected to a vertex in $V_3$ in $T$, which means there are at most $n-a_{23}$ vertices in $V_2$ has degree $2$ in $G^X$. By Claim~\ref{cla:red}, there are at most $n-a_{23}+\alpha(G)$ vertices in $V_1$ has edge in $G^X$. For any isolate vertex in $V_1$ in $T$, it has only one edge in $G'$ connecting to $V_0$, so this vertex must incident on an edge in $T$ which is not in $G'$. So there are at least $n-(n-a_{23}+\alpha(G))=a_{23}-\alpha(G)$ edges incident on $V_1$ in $T$ which is not in $G'$.

    Since any edge has two endpoints, the number of edges in $T$ but not in $G'$ is at least $((n-a_{01})+(a_{01}-\alpha(G))+(n-a_{23})+(a_{23}-\alpha(G)))/2 = n-\alpha(G)$, which means $T$ has cost at least $4n+n-\alpha(G)=5n-\alpha(G)$.
\end{proof}

\begin{corollary}
For any $\eps \in [0, 1/5)$, any algorithm that can estimate $(1,2)$-TSP cost to within a $(1 + \eps)$-factor, can be used to estimate the size of a largest matching in a bipartite graph $G$ on $2n$ vertices to within an additive error of $5 \eps n$.
\end{corollary}

\begin{proof}
    We use the reduction above to construct a $(1,2)$-TSP instance $G'$ on $4n$ vertices. 
    By Lemmas~\ref{lem:red-1} and~\ref{lem:red-2}, the optimal TSP cost for $G'$ 
    is $5n-\alpha(G)$.
    We now run the $(1+\eps)$-approximation algorithm for $(1,2)$-TSP on graph $G'$ (note that the reduction can be simulated in each of neighbor query model, pair query model, and the streaming model without altering the asymptotic number of queries used). Suppose the output is $C$ which satisfies $(1-\eps)(5n-\alpha(G)) \le C \le (1+\eps)(5n-\alpha(G))$, which means $5n-\alpha(G)-5\eps n < C < 5n - \alpha(G) + 5\eps n$. Let $\hat{\alpha} = 5n - C$, we have $\alpha(G)-5\eps n < \hat{\alpha} < \alpha(G)+5\eps n$.
\end{proof}

\section{Additional Lower Bound Results for Approximating Graphic and $(1,2)$-TSP Cost}
\label{sec:additional_lower_bounds}

In this section, we prove several additional lower bounds on approximating the costs of graphic TSP and $(1,2)$-TSP.  Many of these results involve constructing a simple distribution on graphs where some graphs in the support of the distribution have TSP tours of cost close to $n$ while others have cost close to $2n$. We show that no deterministic algorithm can distinguish between these two types of instances, and then invoke Yao's principle~\cite{Yao87} to prove lower bounds for randomized algorithms. When the graphs in the distribution have diameter 2, the graphic TSP instances are also instances of the $(1,2)$ TSP problem. Using this approach we show an $\Omega(n)$ lower bound for both metric TSP and $(1,2)$-TSP costs in our query model.  

In the standard graph query model allowing both pair queries and neighbor queries, we show a stronger lower bound of $\Omega (\eps ^2 n^2)$ for randomized algorithms that estimate the cost of graphic TSP to within a factor of $(2 - \eps )$. This shows that the distance query model is strictly more powerful for estimating graphic TSP cost.

Using Dirac's theorem about the existence of Hamilton Cycles in very dense graphs, we show an $\Omega (n^2)$ lower bound for deterministic algorithms to get any approximation better than 2. For the problem of finding a $(2-\eps)$-approximate tour, rather than just estimating its cost, we show an $\Omega (\eps n^2)$ lower bound for both graphic TSP and $(1,2)$-TSP. 

Finally, we show a space lower bound of $\Omega (\eps n)$ for approximating Graphic TSP to within $2 - \eps$ in the streaming model.

\subsection{An $\Omega(n)$ Query Lower Bound for $(2-\eps)$-approximating $(1,2)$-TSP and Graphic TSP Cost}

In this subsection, we show that in our query model, any randomized algorithm that approximates the cost of minimum graphic TSP or $(1,2)$-TSP to within a factor of $2-\eps$ for any $\eps$, we need $\Omega(n)$ queries. As stated above,  it suffices to create a distribution over $(n+1)$-vertex graphs such that any deterministic algorithm requires $\Omega(n)$ queries to check if the cost of minimum $(1,2)$-TSP or graphic TSP is $n+1$ or $2n$ on this distribution.

The distribution is generated as follows: we start with a ``star'' graph whose vertices set is $\{v_0,v_1,v_2,\dots,v_n\}$ where $v_0$ is connected to all other vertices. Then we pick a random permutation $\pi$ over $[n]$. With probability half, we connect $v_{\pi(i)}$ and $v_{\pi(i+1)}$, for $1 \le i \le n-1$. In this case, the resulting graph is the wheel  graph. With probability half we do not join successive vertices in $\pi$, and the resulting graph is a star graph. Since $v_0$ is connected to all other vertices, any two vertices have distance $1$ or $2$. So graphic TSP and $(1,2)$-TSP are the same in this distribution.

\begin{lemma} \label{lem:st-wh1}
    A wheel graph admits a TSP tour of cost $n+1$ while any TSP tour in a star graph has cost at least $2n$.
\end{lemma}

\begin{proof}
    In a wheel graph, the tour $(v_0,v_{\pi(1)},v_{\pi(2)},\dots,v_{\pi(n)},v_0)$ has cost $n+1$ since all edges are weight $1$. For any tour in a star graph, only the edges incident on $v_0$ have weight $1$. So the cost of the tour is at least $2+2(n-1) = 2n$.
\end{proof}

\begin{lemma} \label{lem:st-wh2}
    If an algorithm only makes $n/3$ queries, then with probability at least $1/3$, the answer to all these queries is the same in a wheel graph and a star graph.
\end{lemma}

\begin{proof} 
    For any query $(v_i,v_j)$, if one of $v_i$ or $v_j$ is $v_0$, then the answer is $1$ in both cases. If none of $v_i$ or $v_j$ is $v_0$, then the answer is $2$ if the graph is the star graph. If the graph is a wheel graph, then the answer is $1$ only if $i$ and $j$ are adjacent to each other in $\pi$, which has probability at most $2/n$. By union bound, with probability at least $1/3$, the answers of all these queries are the same in both cases.
\end{proof}

By Lemma~\ref{lem:st-wh1} and Lemma~\ref{lem:st-wh2}, we have the following lower bound for graphic TSP and $(1,2)$-TSP problem in the distance query model.

\begin{theorem}
    For any $\eps>0$, in the distance query model, any algorithm that with probability at least $2/3$ approximates the cost of $(1,2)$-TSP or graphic TSP within a factor of $(2-\eps)$ requires $\Omega(n)$ queries.
\end{theorem}

\subsection{An $\Omega(\eps^2n^2)$ Query Lower Bound for $(2- 2 \eps)$-approximating Graphic TSP in Standard Graph Query Model}

In this subsection, if an algorithm for graphic TSP is given only access to the underlying graph $G$ via standard graph queries, namely, pair queries, degree queries and neighbor queries, then any randomized algorithm for approximating graphic TSP cost to within a factor of $2-\eps$ for any $\eps > 0$, requires $\Omega(\eps^2n^2)$ queries. Again by Yao's principle, it suffices to create a distribution over $n$-vertex graphs such that any deterministic algorithm requires $\Omega(\eps^2n^2)$ queries to distinguish between graphs where the cost of graphic TSP is $n$ and graphs where the graphic TSP cost is at least $(2-2\eps)n-1$.

We start with a graph $G$ with three parts: a path $P$ with $(1-2\eps)n$ vertices, and two cliques $C_1$ and $C_2$ of size $\eps n$. Let  $u_1, u_2, \dots, u_{\eps n}$  be the vertices in $C_1$, and  $v_1,v_2,\dots,v_{\eps n}$ be the vertices in $C_2$ . Connect all vertices in $C_1$  to an endpoint of $P$, and connect all vertices in $C_2$  to the other endpoint of $P$. For any vertex $u_i$ (resp. $v_i$), we say the $j^{th}$ neighbor of $u_i$ (resp. $v_i$) is $u_j$ (resp. $v_j$) for any $j \neq i$, and the $i^{th}$ neighbor is the endpoint of $P$. For any vertex in $P$, we pick an arbitrary order of its neighbors. With probability half, we change the graph to create a {\sc yes} case as follows: we pick two different indices $i$ and $j$ from $[\eps n]$ randomly. We change the $j^{th}$ neighbor of $u_i$ and $v_i$ to be $v_j$ and $u_j$ respectively and the $i^{th}$ neighbor of $u_j$ and $v_j$ to be $v_i$ and $u_i$ respectively. Otherwise, we do not change the graph and say we are in {\sc no} case.    
\begin{lemma} \label{lem:two-clique1}
    If we are in the {\sc yes} case, then the cost of graphic TSP is $n$. Otherwise, the cost of the graphic TSP is at least $(2-2\eps)n-1$.
\end{lemma}

\begin{proof}
    If we are in the {\sc yes} case, consider the tour that starts at $u_i$, and goes through the vertices in $C_1$ in arbitrary order (but not visiting $u_j$ immediately after $u_i$), then goes to the endpoint of $P$ that connects to all vertices in $C_1$, goes through the path $P$, then visits the vertices in $C_2$ in arbitrary order ending with $v_j$ (but not visiting $v_i$ right before $v_j$), and finally goes back to $u_i$. All edges in this tour have weight $1$. So the cost of this tour is $n$.

    If we are in the {\sc no} case, then all the edges in the path are bridges. So by Lemma~\ref{lem:blo-size}, the cost of graphic TSP is at least $n+(1-2\eps)n-1 = (2-2\eps)n-1$.
\end{proof}

\begin{lemma} \label{lem:two-clique2}
    If an algorithm only makes $\eps^2n^2/4$ queries, then with probability at least $1/3$, the answer to these queries are the same in the {\sc yes} and {\sc no} cases.
\end{lemma}

\begin{proof} 
    The degree of the vertices are the same in both cases. So, any neighbor query has the same answer. For any pair query or neighbor query, all queries on the vertices in $P$ also have the same answer. For any query on the vertices in the cliques, we say a query is querying a pair of indices $(k,\ell)$ if it is a pair query between two vertices with indices $k$ and $\ell$, or it is a neighbor query that queries the $k^{th}$ (resp. $\ell^{th}$) neighbor of $u_{\ell}$ or $v_{\ell}$ (resp. $u_k$ or $v_k$). A pair query or a neighbor query has  different answers in {\sc yes} and {\sc no} cases only when it is querying the indices $i$ and $j$ that we picked when generating the {\sc yes} case. Since $i$ and $j$ are chosen randomly, the probability that a query is querying $i$ and $j$ is $\frac{2}{\eps n(\eps n-1)}$. If the algorithm only make $\eps^2n^2/4$ queries, the probability that there is a query with different answers in {\sc yes} case and {\sc no} case is at most $\frac{\eps^2 n^2}{2\eps n(\eps n -1)} < 2/3$ by union bound.
\end{proof}

By Lemma~\ref{lem:two-clique1} and Lemma~\ref{lem:two-clique2}, we have the following lower bound for graphic TSP problem in the standard query model.

\begin{theorem}
\label{thm:graphic_tsp_lowerbound_standard_model}
    For any $\eps>0$, if an algorithm approximates the cost of graphic TSP within a factor of $(2-\eps)$ with probability at least $2/3$  using only degree queries, neighbor queries and pair queries, then it  requires $\Omega(\eps^2 n^2)$ queries.
\end{theorem}

\subsection{An $\Omega(n^2)$ Query Lower Bound for Deterministic Algorithms for $(1,2)$-TSP and Graphic TSP} \label{sec:det}

In this subsection, we prove that in our stronger, distance query model, any deterministic algorithm that approximates cost of graphic TSP or $(1,2)$-TSP within a factor of $(2-\eps)$  needs $\Omega(\eps n^2)$ queries. 

We first consider the $(1,2)$-TSP problem. We prove that for any $\eps n^2/5$ queries, even if all the answers are that the distance is $2$, the graph could still have a TSP of cost $n+\eps n$.

Consider the graph $H$ whose edge set is pairs of vertices that have not been queried. Since there are only $\eps n^2/5$ queries, there are at least $(1-\eps)n$ vertices that have been queried at most $2n/5 <(1-\eps)n/2 - 1 $ times. These vertices has degree at least $(n-1) - ((1-\eps)n/2 - 1) = (1+\eps)n/2$ in $H$. Let $V_0$ be an arbitrary set that contains exactly $(1-\eps)n$ of these vertices. The subgraph of $H$ induced by $V_0$ has minimum degree at least $(1+\eps)n/2 - \eps n = (1-\eps)n/2 = \card{V_0} / 2$. By the following well-known theorem due to Dirac about the existence of Hamilton cycles in dense graphs, there is a Hamilton cycle in the subgraph of $H$ induced by $V_0$.

\begin{lemma}[Dirac \cite{dirac1952some}]
    Any $n$-vertex graph $G$ where each vertex has degree at least $n/2$ has a Hamilton cycle.
\end{lemma}

So $G$ has a path of length $(1-\eps)n$ that only contains weight one edges, any TSP tour obtained by expanding this path has length at most $(1-\eps)n + 2\eps n = (1+\eps)n$. Thus it is possible that $G$ contains a tour of cost $(1+\eps)n$ after $\eps n^2 /5$ queries.

For graphic TSP problem, we use the same trick as in Section~\ref{sec:ptas_lower_bound}, adding a vertex that connects to all other vertices. This results in  the same lower bound for graphic TSP as for $(1,2)$-TSP.

\begin{theorem}
\label{thm:det_lower_bound}
    Any deterministic algorithm that approximates the cost of graphic TSP or $(1,2)$-TSP to within a factor of $(2-\eps)$ using distance queries needs $\Omega(\eps n^2)$ queries. 
\end{theorem}

\subsection{An $\Omega(\eps n^2)$ Lower Bound for Finding a $(2-\eps)$-Approximate $(1,2)$-TSP or Graphic TSP Tour} \label{sec:solution}

While our focus in this paper has been on estimating the cost of $(1,2)$-TSP or graphic TSP to within a factor that is strictly better than $2$, we show here that if the goal were to output an approximate $(1,2)$-TSP tour or graphic TSP tour (not just an estimate of its cost), then even with randomization, any algorithm requires  $\Omega(\eps n^2)$ distance queries to output a $(2-\eps)$-approximate solution for any $\eps > 0$. We start by showing this lower bound for $(1,2)$-TSP.

We create a distribution over $n$-vertex graphs with $(1,2)$-TSP cost $(1+o(1))n$ such that with a large constant probability, any deterministic algorithm requires $\Omega(\eps n^2)$ queries to output a tour that contains at least $3\eps n$ weight-$1$ edges .  

We generate the graph $G$ with $n$ vertices $\{v_1,v_2,\dots,v_n\}$ as follows: we first generate a random permutation $\pi : [n] \to [n]$. For any $i \neq j$, if $\pi(i) = j$, then $v_i$ and $v_j$ are connected in $G$. 

By construction of $G$, it consists of vertex disjoint cycles, and each cycle in $G$  corresponds to a cycle in permutation $\pi$. Since the expected number of cycles in a random permutation is equal to the $n^{th}$ harmonic number, which is $O(\log n)$ \cite{goncharov1944some}, $G$ has a cycle cover with $O(\log n)$ cycles in expectation. By Markov's inequality, the number of cycles in $G$ is $o(n)$ with probability $1-o(1)$. If we break these cycles into paths and link them in arbitrary order, we obtain a tour of cost at most $n+o(n)$. So the cost of $(1,2)$-TSP of $G$ is $(1+o(1))n$ with probability $1-o(1)$.

Next, we prove that any algorithm needs $\Omega(n)$ queries to find $\eps n$ edges. Construct a graph $H$ that only contains a perfect matching such that the $i^{th}$ vertex on the left is matched to the $j^{th}$ vertex on the right if and only if $\pi(i)=j$. 

Consider the problem of finding the edges in $H$ by pair queries. Each pair query in $G$ can be simulated by at most $2$ pair queries in $H$. Furthermore, any tour in $G$ corresponding to a perfect matching between the vertices in $H$. So to prove the lower bound in $(1,2)$-TSP, we only need to prove that any algorithm that output a perfect matching between the vertices in $H$ contains at most $3\eps n$ edges in $H$. 

The following lemma follows from the arguments in \cite{assadi2019sublinear} (also similar to the arguments in Appendix~\ref{sec:app-lb}) about the lower bound for finding edges in a random perfect matching. 

\begin{lemma} [Section~5.2 in \cite{assadi2019sublinear}] \label{lem:low-mat}
    Any algorithm needs $\Omega(\eps n^2)$ queries to find $\eps n$ edges in a random perfect matching with sufficiently large constant probability.
\end{lemma}

Finally, we prove that if an algorithm only find $\eps n$ edges in $H$, then any output matching contains $\eps n + o(1)$ edges in $H$ with large constant probability. Suppose the algorithm only makes $\eps(1-\eps)(1-2\eps)n^2/3$ queries, then there are at most $\eps(1-\eps)n$ vertices on the left (resp. right) being queried at least $(1-\eps)n/4$ times. Let $V_0$ be the set of vertices $v$ in $H$ such that the edge incident on $v$ is not found by the algorithm and both $v$ and its neighbor are not queried $(1-2\eps)n/3$ times. $V_0$ contains at least $(1-\eps)^2n > (1-2\eps)n$ vertices. Let $H_0$ be the subgraph of $H$ induced by $V_0$. In $H_0$, each vertex $v$ has $\frac{3}{4}\card{V_0}$ vertices $u$ on the other side such that the algorithm does not query the pair $(u,v)$.

By Lemma~\ref{lem:random_matching}, each pair of vertices in $V_0$ contains an edge with probability $O(\frac{1}{n})$. So for any perfect matching between the vertices in $H$, any edge incident on a vertices in $V_0$ is also in $H$ with probability only $O(\frac{1}{n})$. So there are $o(n)$ edges incident on the vertices in $V_0$ that are also  edges in $H$ with probability $1-o(1)$ by Markov's inequality. So the perfect matching contains at most $2\eps n + o(n) < 3\eps n$ edges in $H$ with probability $1-o(1)$, which implies the same lower bound for $(1,2)$-TSP.

For graphic TSP problem, we use the same trick as in Section~\ref{sec:ptas_lower_bound} of adding a vertex that connects to all vertices to prove the same lower bound as for $(1,2)$-TSP.

\begin{theorem} \label{thm:solution}
    Any algorithm that output a graphic TSP or $(1,2)$-TSP tour within a factor of $(2-\eps)$ using distance query with with sufficiently large constant probability needs $\Omega(\eps n^2)$ queries. 
\end{theorem}

\newcommand{\Index}{\ensuremath{\textnormal{\textsf{Index}}}\xspace}

\subsection{An $\Omega(\eps n)$ Lower Bound for $(2-\eps)$ Approximation of Graphic TSP in the Streaming Model}

In this subsection, we prove that any single-pass streaming algorithm that approximates the cost of graphic TSP in insertion-only streams to within a factor of $(2-\eps)$ with probability at least $2/3$ requires $\Omega(\eps n)$ space.

To prove the lower bound for single-pass streaming algorithm, it is sufficient to prove the lower bound in the one-way communication model. The graphic TSP problem in the communication model is the two-player communication problem in which the edge set $E$ of a graph $G(V,E)$ is partitioned between Alice and Bob, and their goal is to approximate the cost of the graphic TSP of $G$. 

We prove the lower bound by a reduction from the \Index problem, In \Index, Alice is given a bit-string $x \in \set{0,1}^{N}$, Bob is given an index $k^{\star} \in [N]$, and the goal is for Alice to send a message to Bob so that Bob outputs $x_{k^{\star}}$. It is well-known that any one-way communication protocal that solves \Index with probability 2/3 requires $\Omega(N)$ bits of communication~\cite{kremer1999randomized}.

We use the \Index problem with size $N = \eps n/4$. We will construct a graph $G$ such that the cost of graphic TSP is at most $n+2N$ if $x_{k^{\star}}=1$ and at least $2n-N-1$ if $x_{k^{\star}}=0$. Since $N=\eps n /4$, in order to approximate the cost of graphic TSP within a factor of $(2-\eps)$, Alice and Bob need to be able to check if the cost of graphic TSP is larger than $2n-2N$ or less than $n+2N$.

\subsubsection*{Reduction:} 
Given an instance of \Index with size $N=\eps n/4$:
\begin{enumerate}
    \item Alice and Bob construct the following graph $G(V,E)$ with no communication: The vertex set $V$ is a union of three set $P$, $U$, $W$, where $P=\{v_1,v_2,\dots,v_{n-2N}\}$, $U=\{u_1,u_2,\dots,u_N\}$ and $W=\{w_1,w_2,\dots,w_N\}$. In Alice's graph, all vertices in $P$ form a path, whose endpoints are $v_1$ and $v_{n-2N}$, for any $i \in [N]$, there is an edge between $u_i$ and $w_i$. Furthermore, $v_1$ connects to all $u_i$ such that $x_i = 1$ in her input in the \Index instance. In Bob's graph, $v_1$ connects to all $w_i$ for $i \neq k^{\star}$ and $w_{k^{\star}}$ connects to $v_{n-2N}$ instead.
    \item Alice and Bob then approximate the cost of graphic TSP of the graph $G$ using the best protocol. Bob outputs $x_{k^{\star}} = 0$ if the cost of graphic TSP is larger than $2n-2N$ and outputs $x_{k^{\star}} = 1$ otherwise.
\end{enumerate}

The communication cost of this protocol is at most as large as the communication complexity of the protocol used to solve graphic TSP. Now we prove the correctness of the reduction.

\begin{lemma} \label{lem:lb-stream}
    If $x_{k^{\star}}=1$, then the cost of graphic TSP of $G$ is at most $n+2N$. If $x_{k^{\star}}=0$, then the cost of graphic TSP of $G$ is at least $2n-N-1$.
\end{lemma}

\begin{proof}
    If $x_{k^{\star}}=1$, consider the tour that first visits the path in $P$ from $v_1$ to $v_{2n-2N}$, then visits $w_{k^{\star}}, u_{k^{\star}}$, then visits $w_i,u_i$ for each $i \neq k^{\star}$ in arbitrary order, and finally goes back to $v_1$. Since $x_{k^{\star}}= 1$, $u_{k^{\star}}$ connects to $v_1$. Also for each $i \neq k^{\star}$, $w_{k^{\star}}$ connects to $v_1$, so for any $i$ and $j \neq k^{\star}$, $u_i$ and $w_j$ have distance at most $3$. Since any other edge in the tour has weight $1$, the cost of the tour is at most $n+(3-1) \cdot N = n+2N$.

    If $x_{k^{\star}}=0$, both $u_{k^{\star}}$ and $w_{k^{\star}}$ do not connect to $v_1$. So $u_{k^{\star}}$, $w_{k^{\star}}$ and $v_{n-2N}$ form a block in $G$. For any $i \neq k^{\star}$, both $u_i$ and $w_i$ do not connect to $v_{n-2N}$. So $u_i$, $w_i$ and $v_1$ forms a block in $G$. Furthermore, all edges in the path are bridges in $G$. By Lemma~\ref{lem:blo-size}, the cost of graphic TSP of $G$ is at least $n+N+(n-2N)-1 = 2n-N-1$.
\end{proof}

\begin{theorem}
    For any $\eps>0$, any single-pass streaming algorithm that is able to approximate the cost of graphic TSP of an input graph $G$ within a factor of $2-\eps$ in insertion-only streams with probability at least $2/3$ requires $\Omega(\eps n)$ space.
\end{theorem}

\begin{proof}
    Let $\Pi$ be any $1/3$-error one-way protocol  that approximates  graphic TSP within a factor of $(2-\eps)$. By Lemma~\ref{lem:lb-stream}, we obtain a protocol for \Index that errs with probability at most $1/3$ and has communication cost at most equal to cost of $\Pi$. By the $\Omega(N)$ lower bound on the one-way communication complexity of \Index, we obtain that communication cost of $\Pi$ must be $\Omega(N) = \Omega(\eps n)$. The theorem now follows from this argument, as one-way communication complexity lower bounds the space complexity of single pass streaming algorithms. 
\end{proof}

\clearpage

\bibliographystyle{abbrv}
\bibliography{general}

\appendix

\section{Approximating Maximal Matching and Maximal Matching Pair with $\tilde{O}(n^{1.5})$ Pair Queries}

In this section, we give algorithms that approximates the size of a maximal matching and a maximal matching pair within a factor of $(1+\eps)$ with $\tilde{O}(n^{1.5}/\eps^2)$ pair queries. Both algorithms are built on the algorithm in \cite{yoshida2012improved} that approximates the size of a maximal independent set. 

We first describe the algorithm and result in \cite{yoshida2012improved}. Given a graph $G$ with $n$ vertices and $m$ edges, consider the following process that generates a maximal independent set: pick a random permutation $\pi$ on all vertices. Maintain a set $S$, initially empty. Consider each vertex in turn, from the lowest rank to the highest rank. For any vertex $v$, if $S$ contains no neighbor of $v$, then add $v$ to $S$. The algorithm $IO^{\pi}(v)$ (Algorithm~\ref{alg:I_O}) checks if a vertex $v$ is inside the maximal independent set generated by $\pi$. 

\begin{algorithm}[h] 
    \If {$IO_G^{\pi}(v)$ has already been computed}{\Return the computed answer.}
    Let $v_1,v_2,\dots,v_t$ be the neighbors of $v$, in order of increasing rank.\;
    $i \leftarrow 1$.\;
    \While {$\pi(v_i) < \pi(v)$}{
        \If {$IO^{\pi}(v_i) =$~{\sc true}}{\Return {\sc false}} 
        $i \leftarrow i+1$.
    } \Return {\sc true} 
    \caption{ $IO^{\pi}_G(v)$: Check if a vertex $v$ is in the maximal independent set \cite{yoshida2012improved}} \label{alg:I_O}
\end{algorithm}

Let $T_G(\pi,v)$ be the number of calls to $IO^{\pi}_G$ when calling $IO^{\pi}_G(v)$. The following lemma gives an upper bound on the expected value of $T_G(\pi,v)$ when $v$ and $\pi$ are chosen randomly. 

\begin{lemma} [Theorem~2.1 in \cite{yoshida2012improved}] \label{lem:io}
    For any graph $G$ with $n$ vertices and $m$ edges, 
    $$
    \mathbb{E}_{v,\pi}[T_G(\pi,v)] \le 1 + \frac{m}{n}
    $$
\end{lemma}

\subsection{Approximating Maximal Matching (Proof of Theorem~\ref{lem:max-mat})} \label{sec:app-mat}
Given a graph $G=(V,E)$, let $L(G)$ be the line graph of $G$, where the vertices in $L(G)$ are the edges in $G$, and two vertices in $L(G)$ are neighbors if they share a common endpoint in $G$. Suppose $d^G_{\max}$ is the maximum degree in $G$, then any vertex in $L(G)$ has degree at most $2d^G_{max}-2$. Furthermore, any maximal independent set in $L(G)$ is a maximal matching in $G$. To approximate a maximal matching in $G$, it is sufficient to approximate a maximal independent set in $L(G)$.

By Lemma~\ref{lem:io}, $\mathbb{E}_{e,\pi}[T_{L(G)}(\pi,e)] = O(d^G_{\max})$, if an edge $e$ in $G$ (also a vertex in $L(G)$) and a permutation $\pi$ on the edges in $G$ are chosen randomly. However, $d^G_{\max}$ can be as large as $n$ and the number of edges in $G$ can be as large as $n^2$. Thus, to approximate the maximal matching within an additive error $\eps n$, we need to sample $\Omega(n)$ edges in $G$ and check if each one is in the maximal matching by calling $IO^{\pi}_{L(G)}(e)$. So the total number of calls to $IO^{\pi}_{L(G)}$ can be as large as $n^2$.

To reduce the number of calls to $IO^{\pi}$, we design a two-phase algorithm that approximates the size of a maximal matching. In the first phase, we match high degree vertices greedily. In the second phase, we use the process described earlier to approximate the size of a maximal matching in the remaining low-degree graph after the first phase.

In the first phase, we run Algorithm~\ref{alg:mat-high}, which returns a partial matching $M$ and a vertex set $S$ that contains all vertices not matched in $M$. The algorithm works as follows: at first, $M$ is empty and $S$ is the vertex set $V$. We consider all vertices one by one in arbitrary order. When considering $v$, if $v$ is not matched in $M$, then we sample $c_0=100\sqrt{n}\log n$ vertices from $S$, and check if some eighbor of $v$ is among them. If so, we add $v$ and one of its neighbors into $M$ and delete these two vertices from $S$. The algorithm only uses $\tilde{O}(n^{1.5})$ pair queries since for each vertex we only check if it is a neighbor of $\tilde{O}(\sqrt{n})$ vertices. We prove that with high probability, the subgraph of $G$ induced by $S$ has degree at most $\sqrt{n}$. 

\begin{algorithm}[h] 
    $S \leftarrow V$, $M \leftarrow \emptyset$, $c_0 \leftarrow 100\sqrt{n} \log n$.\;
    \For {$v \in V$} {
        \If {$v \notin S$} {
            \textbf{continue}\;
        }
        Sample $c_0$ vertices $u_1,u_2,\dots,u_{c_0}$ from $S$.\;
        \For {$i \leftarrow 1 \ldots c_0$}{
            \If {$(v,u_i)$ is an edge in $G$} {
                $M \leftarrow M \cup \{(v,u_i)\}$, $S \leftarrow S \setminus \{v,u_i\}$ \;
                \textbf{break}\;
            }
        }
    }
    \textbf{Output} $S$ and $M$.
    \caption{Match high degree vertices} \label{alg:mat-high}
\end{algorithm}

\begin{lemma} \label{lem:mat-p1}
    The subgraph of $G$ induced by $S$ has degree at most $\sqrt{n}$ with probability at least $1-o(\frac{1}{n})$.
\end{lemma}

\begin{proof}
    %For any vertex $v$ that has at least $\sqrt{n}$ neighbors in $S$, it also has at least $\sqrt{n}$ neighbors in $S$ during the time when the algorithm is running the loop on $v$. On the other hand,
     $v$ is unmatched in $M$ only if when running the loop on $v$, none of the  $c_0 = 100\sqrt{n} \log n$ vertices sampled are neighbors of $v$. The probability that this happens is at most $o(\frac{1}{n^2})$ by Chernoff bound. Taking the union bound on all vertex $v$, with probability at least $1-o(\frac{1}{n})$, all vertices in $S$ have at most $\sqrt{n}$ neighbors in $S$.
\end{proof}

After the first phase, we only need to approximate the size of a maximal matching in the subgraph $G_S$ of $G$ induced by $S$, which has degree at most $\sqrt{n}$. There are at most $n\sqrt{n}$ edges in $G_S$. So in order to approximate the size of a maximal matching  we only need to sample $O(\sqrt{n})$ edges  and run $IO^{\pi}_{L(G_S)}$ on them. The expected total number of calls to $IO^{\pi}_{L(G_S)}$ is $O(\sqrt{n} \cdot \sqrt{n}) = O(n)$ by Lemma~\ref{lem:io} since the maximum degree in $G_S$ is $\sqrt{n}$. The following lemma shows that we can simulate this process using $\tilde{O}(n^{1.5})$ pair queries.

\begin{lemma} \label{lem:mat-p2}
    There is an algorithm that approximates the size of a maximal independent set in $L(G_S)$ within an additive error $\eps n$ that uses $\tilde{O}(n^{1.5}/\eps^2)$ pair queries with probability at least $2/3$.
\end{lemma}

We first prove an auxiliary claim.

\begin{claim} \label{cla:ranking}
    Suppose there are two sets of objects $S_1$ and $S_2$ such that $\frac{\card{S_1}}{\card{S_2}} < \sqrt{n}$ and $\card{S_1}+\card{S_2} < n^2$. If we pick a rank permutation $\pi$ on all objects in $S_1 \cup S_2$, then with probability $1-o(\frac{1}{n^2})$, any set of  $10\sqrt{n}\log n$ successive objects in $\pi$ contains an object in $S_2$.
\end{claim}

\begin{proof}
    Let $m_0=\card{S_1}+\card{S_2}$ and $m_1=\card{S_1}$.
    For any set of $10\sqrt{n}\log n$ indices, the probability that there is no object in $S_2$ on these indices in $\pi$ is 
    $$\frac{m_1}{m_0} \cdot \frac{m_1-1}{m_0-1} \dots \frac{m_1-10\sqrt{n}\log n +1}{m_0-10\sqrt{n}\log n+1} < (\frac{m_1}{m_0})^{10\sqrt{n}\log n} < (1-\frac{1}{2\sqrt{n}})^{10\sqrt{n}\log n} = o(\frac{1}{n^4})$$ 
    By taking the union bound on all possible successive indices, the probability that any successive $10\sqrt{n}\log n$ objects in rank $\pi$ contains an object in $S_2$ is at least $1-o(\frac{1}{n^2})$.
\end{proof}

\begin{figure}[h!]
\centering
\begin{tikzpicture}
	\node[ellipse, black, fill=black!5, draw, line width=1pt, minimum width=160pt, minimum height=80pt] (W) {$W$};
	\node[ellipse, black, fill=black!10, draw, line width=1pt, minimum width=120pt, minimum height=60pt] (D1) [right = 30pt of W] {$D_1$};
	\node[ellipse, black, fill=black!10, draw, line width=1pt, minimum width=50pt, minimum height=40pt] (LG) [left = 40pt of D1] {$L(G_S)$};
	\node[ellipse, black, fill=black!10, draw, line width=1pt, minimum width=50pt, minimum height=30pt] (D2) [above left = 30pt and 10pt of D1] {$D_2$};
	\node (A1) [above left = 35pt and 10pt of D1] {};
	\node (A2) [left = 20pt of A1] {};
	\node (B1) [left = 45pt of D1] {};
	\node (B2) [left = 30pt of B1] {};
	\node (C1) [right = 50pt of W] {};
	\node (C2) [below right = 10pt and 30pt of C1] {};
	\node (C3) [right = 50pt of C1] {};

	\draw (A1.center) to (B1.center);
	\draw (A1.center) to (B2.center);
	\draw (A2.center) to (B1.center);
	\draw (A2.center) to (B2.center);
	\draw (A1.center) to (C1.center);
	\draw (A1.center) to (C2.center);
	\draw (A1.center) to (C3.center);
	\draw (A2.center) to (C1.center);
	\draw (A2.center) to (C2.center);
	\draw (A2.center) to (C3.center);

\end{tikzpicture}

\caption{Illustration of $H$} \label{fig:mat-app}

\end{figure}
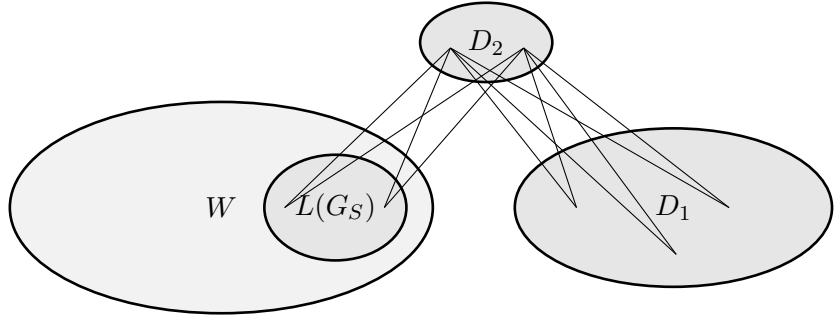

\begin{proof} [Proof of Lemma~\ref{lem:mat-p2}]
    Suppose $E_S$ is the set of edges of $G_S$. Let $W$ be the set of  pairs of vertices in $G$. By definition, $\card{W} = n(n-1)/2$. We construct a graph $H$ as follows: the vertex set of $H$ is $W \cup D_1 \cup D_2$ where $\card{D_1} = n \sqrt{n}$ and $\card{D_2} = \sqrt{n}$. Each vertex in $D_2$ is connected to each vertex in $D_1 \cup E_S$. For any two vertices in $E_S$, they are connected if and only if they have a common endpoint in $G_S$. All vertices in $W \setminus E_S$ are isolated vertices. See Figure~\ref{fig:mat-app} as an illustration.
    
    Let $\pi$ be a random permutation on the vertices in $H$. If the lowest rank among $E_S \cup D_1 \cup D_2$ is a vertex in $E_S \cup D_1$, which has probability at least $\frac {\card{D_1}}{\card{D_1}+\card{D_2}} > 1 - \frac{1}{n}$, no vertex in $D_2$ is inside the maximal independent set of $H$ generated by $\pi$. In this case, the size of maximal independent set is the size of a maximal independent set of $L(G_S)$ plus $\card{D_1} = n \sqrt{n}$ plus the size of $W \setminus E_S$. Thus, to prove the lemma, it is sufficient to approximate the size of $E_S$ and the size of a maximal independent set of $H$ both with additive error $\eps n /2$.

    To approximate the size of $E_S$, we sample $M_0 = \frac{100 n \sqrt{n}}{\eps^2}$ pairs of vertices in $G$ and query if there is an edge between them. Suppose $\bar{X_0}$ sampled pairs are edges in $G_S$, then by Chernoff bound, $\card{\frac{\bar{X_0} \cdot \card{W}}{M_0} - \card{E_S}} < \eps n /2$ with probability at least $1-\frac{1}{2^{\sqrt{n}}} = 1-o(\frac{1}{n})$ since $\card{W}<n^2$. Thus, we can approximate the size of $E_S$ with additive error $\eps n /2$ by making $O(n\sqrt{n}/\eps^2)$ pair queries with probability $1-o(\frac{1}{n})$.

    Let $V_H$ be the set of vertices in $H$. By definition, $\card{V_H} = n(n-1)/2 + n\sqrt{n} + \sqrt{n}$. We use Algorithm~\ref{alg:I_O} to approximate the size of the maximal independent set of $H$ generated by $\pi$. We sample $M_1 = \frac{100 n}{\eps^2}$ vertices in $V_H$, and call $IO^{\pi}_H$ on them. Suppose $IO^{\pi}_H$ returns {\sc true} $\bar{X_1}$ times and the size of maximal independent set of $H$ generated by $\pi$ is $X_1$, then by Chernoff bound, $\card{\frac{\bar{X_1} \cdot \card{V_H}}{M_1} - X_1} < \eps/2$ with probability $9/10$ since $\card{V_H}<n^2$. 

    Now we bound the number of calls to $IO^{\pi}_H$. The number of vertices in $H$ is $n(n-1)/2+n\sqrt{n}+\sqrt{n} = \Omega(n^2)$. Since $G_S$ has degree at most $\sqrt{n}$, $L(G_S)$ also has degree $O(\sqrt{n})$ and $\card{E_S} < n\sqrt{n}$. So the number of edges in $L(G_S)$ and the number of edges incident on $D_2$ are both $O(n^2)$, which means $H$ has $O(n^2)$ edges. By Lemma~\ref{lem:io}, if we randomly choose a vertex $v \in V_H$, $\mathbb{E}_{v,\pi}[T_H(v,\pi)] = O(1)$. So the total number of calls to $IO^{\pi}_H$ is $O(n / \eps^2)$ in expectation.

    Finally, we describe how to simulate $IO^{\pi}_H(v)$ by pair queries in $G$, and bound the number of queries. There are three kinds of vertices in $H$, the vertices in $D_1$, $D_2$ or $W$. If $v \in D_1$, the neighbors of $v_1$ are all vertices in $D_2$. We do not need any query to figure out the neighbors of $v$. If $v \in D_2$, then the neighbors of $v$ are all vertices in $E_S \cup D_1$. We simulate $IO^{\pi}_H(v)$ as follows: we consider all vertices in $W \cup D_1$ one by one from lowest rank to highest rank. When considering vertex $u$, if the rank of $u$ is larger than $v$ then return {\sc true}, otherwise we use at most one pair query to check if $u \in E_S \cup D_1$, if so, we run $IO^{\pi}_H(u)$, otherwise do nothing and continue to the next $u$. Since $\card{W \cup D_1} < n^2$ and $\card{E_S \cup D_1} \ge \card{D_1} = n\sqrt{n}$, by Claim~\ref{cla:ranking}, with probability $1-o(\frac{1}{n^2})$, we use at most $10\sqrt{n}\log n$ pair queries between successive recursive calls to $IO^{\pi}_H$. If $v \in W$, we first use one pair query to check if $v \in E_S$. If $v \notin E_S$, we can directly output {\sc true}. If $v \in E_S$, let $v_i$ and $v_j$ be the endpoints of $v$ in $G_S$. We consider all vertices in $D_2$ and all vertices in $W$ that contains either $v_i$ or $v_j$ (there are $2n-1$ of them). When considering vertex $u$, if the rank of $u$ is larger than $v$ then return {\sc true}. Otherwise, we use at most one pair query to check if $u \in E_S \cup D_1$, if so, we run $IO^{\pi}_H(u)$. Otherwise do nothing and continue to the next $u$. Since there are at most $2n-1 + \sqrt{n}$ vertices we considered and at least $\card{D_2}=\sqrt{n}$ of them are $v$'s neighbors, the number of pair queries between successive calls to $IO^{\pi}_H$ is at most $10\sqrt{n}\log n$ with probability $1-o(\frac{1}{n^2})$ by Claim~\ref{cla:ranking}. By taking union bound on all $v$, we use $\tilde{O}(n^{1.5} / \eps^2)$ pair queries in total in expectation with probability $1-o(1)$. By Markov's inequality, we use $\tilde{O}(n^{1.5}/ \eps^2)$ queries in total with probability at least $9/10$.

    The failure probability is at most $1/10+\frac{1}{n}+o(\frac{1}{n}) < 1/3$ by union bound.
\end{proof}

Note that the algorithm in Lemma~\ref{lem:mat-p2} needs to sample the rank of all possible vertex pairs. So it requires $O(n^2)$ time although it only uses $\tilde{O}(n^{1.5})$ queries. However, we can use the ``permutation generation on the fly'' technique in \cite{onak2012near} (see Section~4 in \cite{onak2012near} for details) to generate the rank of a pair of vertices only when the algorithm needs it. By doing this, the algorithm in Lemma~\ref{lem:mat-p2} can work in $\tilde{O}(n^{1.5})$ time.

\begin{lemma} \label{lem:mat-p3}
    There is an algorithm that approximates the size of a maximal independent set in $L(G_S)$ within an additive error $\eps n$ that uses $\tilde{O}(n^{1.5}/\eps^2)$ time with probability at least $2/3$.
\end{lemma}

Theorem~\ref{lem:max-mat} follows from Lemma~\ref{lem:mat-p1} and Lemma~\ref{lem:mat-p3}.

\subsection{Approximating Maximal Matching Pair (Proof of Theorem~\ref{lem:mat-pair})} \label{sec:mat-pair}
In this section, we generalize the idea in Section~\ref{sec:app-mat} to an algorithm that approximates a maximal matching pair in a graph $G$. The algorithm has two phases. In the first phase, we greedily match the high degree vertices. In the second phase, we construct an auxiliary graph $H$ such that approximation of a maximal independent set of $H$ gives an approximation of a maximal matching pair of the graph $G$. 

In the first phase, we run Algorithm~\ref{alg:mat-pair}, which returns two partial matchings $M_1$, $M_2$ such that $M_1 \cap M_2 = \emptyset$, and two sets $S_1$, $S_2$ that contain all vertices not matched in $M_1$ and $M_2$ respectively. The algorithm works similarly to Algorithm~\ref{alg:mat-high}. For any vertex $v$ and $j\in \{1,2\}$, if $v \in S_j$,  we sample $c_0 = 100\sqrt{n}\log n$ vertices from $S_j$ and check if there is a neighbor of $j$. Suppose there is a sampled $u$ such that $(u,v)$ is an edge and it is not in in $M_{3-j}$. Then we add this edge into $M_j$ and delete these two vertices from $S_j$. Like Algorithm~\ref{alg:mat-high}, Algorithm~\ref{alg:mat-pair} also only uses $\tilde{O}(n^{1.5})$ queries. We prove that the subgraphs of $G$ induced by $S_1$ and $S_2$ both have maximum degree at most $\sqrt{n}$ with high probability.

\begin{algorithm}[h] 
    $S_1 \leftarrow V$, $S_2 \leftarrow V$, $M_1 \leftarrow \emptyset$, $M_2 \leftarrow \emptyset$, $c_0 \leftarrow 100\sqrt{n} \log n$.\;
    \For {$v \in V$} {
        \If {$v \in S_1$} {
            Sample $c_0$ vertices $u_1,u_2,\dots,u_{c_0}$ from $S_1$.\;
            \For {$i \leftarrow 1 \ldots c_0$}{
                \If {$(v,u_i)$ is an edge in $G$ and $(v,u_i) \notin M_2$} {
                    $M_1 \leftarrow M_1 \cup \{(v,u_i)\}$, $S_1 \leftarrow S_1 \setminus \{v,u_i\}$ \;
                    \textbf{break}\;
                }
            }
        }
        \If {$v \in S_2$} {
            Sample $c_0$ vertices $u_1,u_2,\dots,u_{c_0}$ from $S_2$.\;
            \For {$i \leftarrow 1 \ldots c_0$}{
                \If {$(v,u_i)$ is an edge in $G$ and $(v,u_i) \notin M_1$} {
                    $M_2 \leftarrow M_2 \cup \{(v,u_i)\}$, $S_2 \leftarrow S_2 \setminus \{v,u_i\}$ \;
                    \textbf{break}\;
                }
            }
        }
    }
    \textbf{Output} $S_1$, $S_2$, $M_1$ and $M_2$.
    \caption{Match high degree vertices into a matching pair} \label{alg:mat-pair}
\end{algorithm}

\begin{lemma} \label{lem:mp-p1}
    The subgraphs of $G$ induced by $S_1$ and $S_2$ both have degree at most $\sqrt{n}$ with probability at least $1-o(\frac{1}{n})$.
\end{lemma}

The proof of Lemma~\ref{lem:mp-p1} is similar to the proof of Lemma~\ref{lem:mat-p1} and we omit it here.

In the second phase, we construct a graph $H_P$ such that any maximal independent set in $H_P$ with $M_1$ and $M_2$ represent a maximal matching pair of $G$. Let $G_{S_1}$ be the subgraph of $G$ induced by $S_1$ without the edges in $M_2$, and $G_{S_2}$ be the subgraph of $G$ induced by $S_2$ without the edges in $M_1$. Let $L(G_{S_1})$ and $L(G_{S_2})$ be the line graph of $G_{S_1}$ and $G_{S_2}$ respectively. The graph $H_P$ contains a copy of $L(G_{S_1})$ and a copy of $L(G_{S_2})$; furthermore, for any pair of vertices in $L(G_{S_1})$ and $L(G_{S_2})$ that represent the same edge in $G$, we add an edge between them. 

\begin{lemma} \label{lem:hp}
    For any maximal independent set $I$ in $H_P$, suppose $I_1$ is the set of edges in $G$ that are represented by a vertex in $I$ inside the copy of $L(G_{S_1})$ and $I_2$ is the set of edges in $G$ that represented by a vertex in $I$ inside the copy of $L(G_{S_2})$. $(M_1 \cup I_1 , M_2 \cup I_2)$ is a maximal matching pair of $G$.
\end{lemma}

\begin{proof}
    Since $L(G_{S_1})$ and $L(G_{S_2})$ are line graphs, $I_1$ and $I_2$ are both matchings. Furthermore, any pair of vertices in $H$ that represent the same edge have an edge in $H$ between them, so $I_1 \cap I_2 = \emptyset$. Additionally, $G_{S_1}$ and $G_{S_2}$ contains no edges in $M_1 \cup M_2$, $G_{S_1}$ (resp. $G_{S_2}$) contains no vertices in $M_1$ (resp. $M_2$). $M_1 \cup I_1$ and $M_2 \cup I_2$ are both matchings and $(M_1 \cup I_1) \cap (M_2 \cup I_2) = \emptyset$, which means $(M_1 \cup I_1,M_2 \cup I_2)$ is a matching pair of $G$. 

    For any edge $e$ in $G$ which is not in $M_1 \cup I_1 \cup M_2 \cup I_2$, we prove that there are two edges $e_1 \in M_1 \cup I_1$ and $e_2 \in M_2 \cup I_2$ that each share an endpoint with $e$. If $e$ is not an edge in $G_{S_1}$, then $e$ contains an endpoint that is matched in $M_1$ since $e \notin M_2$, which means there is an edge $e_1 \in M_1$ that shares a common endpoint with $e$. If $e$ an edge in $G_{S_1}$, suppose $v^e_1$ is the vertices in $H$ that represents $e$ in the copy of $L(G_{S_1})$. Since $e \notin I_1$, $v^e_1 \notin I$. There is a vertex $u \in I$ that connects to $v^e_1$ in $H$. Since $e \notin M_2$, $u$ is not in the copy of $L(G_{S_2})$, which means $u$ is inside the copy of $L(G_{S_1})$. Suppose $e_1$ is the edge represented by $u$, by definition of $I_1$ and $L(G_{S_1})$, $e_1 \in I_1$ and $e_1$ shares an endpoint with $e$. 

    Following a similar argument, there is also an edge $e_2 \in {M_2 \cup I_2}$ that shares an endpoint with $e$. 
\end{proof}

By Lemma~\ref{lem:hp}, to approximate the size of a maximal matching pair of $G$, it is sufficient to approximate the size of a maximal independent set of $H_P$. The following lemma follows from a similar argument as the prove of Lemma~\ref{lem:mat-p2} and Lemma~\ref{lem:mat-p3}.

\begin{lemma} \label{lem:mp-p2}
    There is an algorithm that approximates the size of a maximal independent set in $H_P$ within an additive error $\eps n$ that uses $\tilde{O}(n^{1.5}/\eps^2)$ pair queries and time with probability at least $2/3$.
\end{lemma}

Theorem~\ref{lem:mat-pair} follows from Lemma~\ref{lem:mp-p1} and Lemma~\ref{lem:mp-p2}.

\iffalse
\begin{theorem}
    For any graph $G$ with $n$ vertices, we can approximate the size of a maximal matching pair to within an additive error of $\eps n$ using $\tilde{O}(n^{1.5}/ \eps^2)$ pair queries with probability at least $2/3$.
\end{theorem}
\fi

\end{document}